\documentclass[11pt,draftcls,onecolumn]{IEEEtran}

\usepackage{cite}

\usepackage{colortbl}
\usepackage{color}
\usepackage{xcolor}
\usepackage{bm}
\usepackage{bbm}
\usepackage{amsmath}
\usepackage{multirow}

\usepackage{cases}

\usepackage{amssymb}
\usepackage{amsthm}
\usepackage{graphicx}
\usepackage{mathrsfs}
\usepackage{empheq}	
\usepackage[mathcal]{euscript}
\usepackage[margin = 2cm]{geometry}
\usepackage{framed}
\usepackage{comment}

\usepackage{dsfont}
\usepackage{xr}
\makeatletter
\newcommand*{\addFileDependency}[1]{
  \typeout{(#1)}
  \@addtofilelist{#1}
  \IfFileExists{#1}{}{\typeout{No file #1.}}
}
\makeatother

\DeclareFontFamily{U}{mathx}{\hyphenchar\font45}
\DeclareFontShape{U}{mathx}{m}{n}{
      <5> <6> <7> <8> <9> <10>
      <10.95> <12> <14.4> <17.28> <20.74> <24.88>
      mathx10
      }{}
\DeclareSymbolFont{mathx}{U}{mathx}{m}{n}
\DeclareFontSubstitution{U}{mathx}{m}{n}
\DeclareMathAccent{\widecheck}{0}{mathx}{"71}

\newtheorem{theorem}{Theorem}
\newtheorem{lemma}{Lemma}
\newtheorem{corollary}{Corollary}

\newtheorem{assumption}{Assumption}

\newtheorem{definition}{Definition}

\newcommand{\E}{\mathbb{E}}

\newcommand{\beq}{\begin{equation}}
\newcommand{\eeq}{\end{equation}}
\newcommand{\beqa}{\begin{IEEEeqnarray}{rCl}}
\newcommand{\eeqa}{\end{IEEEeqnarray}}

\newcommand{\thetatrue}{\theta_0}

\DeclareMathOperator*{\argmin}{arg\,min}

\title{A Fundamental Limit in Decentralized Decision-Making
\thanks{
M. Carpentiero, F. Scala, and V. Matta are with the Department of Information and Electrical Engineering and Applied Mathematics (DIEM), University of Salerno, I-84084 Fisciano (SA), Italy, and also with the National Inter-University Consortium for Telecommunications (CNIT), Italy (e-mails: \{mcarpentiero, fescala, vmatta\}@unisa.it).

A. H. Sayed is with
the Institute of Electrical and Micro Engineering, 
EPFL, CH-1015 Lausanne, Switzerland (e-mail: ali.sayed@epfl.ch).

}
}

\author{\IEEEauthorblockN{Marco Carpentiero, Felice Scala, Vincenzo Matta, and Ali H. Sayed}
}

\begin{document}
\maketitle

\begin{abstract}
In decentralized decision-making, several agents connected according to a network graph aim at solving a classification problem by collecting streaming observations. Due to decentralization, they run an iterative algorithm where, at each iteration, they can only exchange information locally with their neighbors. 
While decentralized estimation solutions have been shown to match the performance of optimal centralized systems, we show here that surprisingly this conclusion does not hold for decentralized decision-making.
Specifically, we prove that the error probability for the \emph{best decentralized} decision strategy exhibits \emph{an irreducible loss} with respect to the optimal centralized classifier. 
This result establishes a fundamental limit for the performance of any decentralized decision strategy.
We obtain an analytical relation showing that this limit is related to the interplay between  decentralization and classification. The first aspect appears through the distances between the nodes in the graph, while the second aspect plays through the moment generating functions of the likelihood ratios that describe the decision problem. By applying the derived closed-form relation to different network topologies and inference problems, we observe some interesting and perhaps unexpected behavior emerging.  
In particular, we characterize the scaling law (with the network size) for the loss over popular network topologies, showing that the error probabilities might differ by orders of magnitude; and we examine how performance is affected by the relative distance between informative and uninformative agents over the graph.
\end{abstract}

\begin{IEEEkeywords}
Social learning, error probability, large deviations, exact asymptotics.
\end{IEEEkeywords}

\section{Introduction and Related Work}
The problem of distributed decision-making refers to a collection of spatially dispersed agents that cooperate to accomplish a classification task. Each agent collects some data and interacts with the other agents according to some network structure. 

Many earlier works on distributed decision-making focused on architectures with a fusion center that can be reached by any agent~\cite{Veeravalli,Poor,Varshney,Tong,Willett}. 
Later on, the rising popularity of \emph{fully flat or decentralized} architectures led to the following paradigm shift: $i)$ each agent acts as a fusion center, meaning that \emph{it is responsible for making its own decisions}; and $ii)$ the decisions of the individual agents leverage the information disseminated across the network by means of \emph{local consultation steps among neighbors}. 
This fully decentralized paradigm is also referred to as \emph{decentralized decision-making} or \emph{social learning}~\cite{MattaBordignonSayedBook,SocLearnSPmagazine}. 

There exist several useful strategies to perform decentralized decision-making. All of them are characterized by the essential constraint of decentralization: at each round of the inferential algorithm (i.e., at each time epoch), each agent is allowed to receive information only from its neighbors, according to the graph topology that defines the network. 
Under this constraint, some schemes aggregate the beliefs (i.e., the posterior probabilities) of the neighbors by means of a weighted arithmetic average~\cite{zhaoLearningSocialNetworks2012, jadbabaieNonBayesianSocialLearning2012}; while other schemes use a weighted \emph{geometric} average of the posterior probabilities~\cite{lalithaSocialLearningDistributed2018,nedicFastConvergenceRates2017, Jadbabaie2018}. One fundamental property that has been shown for both arithmetic and geometric averaging is \emph{consistency}: under mild conditions on the network connectivity and the identifiability of the classification problem, decentralized decision-making attains a vanishing error probability as time increases~\cite{MattaBordignonSayedBook}. 

Since consistency is an \emph{asymptotic} property, different consistent strategies may exhibit very different error probabilities for finite time horizons. 
Therefore, to properly compare the decentralized decision strategies, it is necessary to go beyond consistency and provide a quantitative evaluation for the error probabilities. 
The main works addressing this issue focused on a \emph{large deviation analysis}~\cite{MattaBordignonSayedBook,MouraLDnonGauss,lalithaSocialLearningDistributed2018,nedicFastConvergenceRates2017}. The takeaway from these works is that (under geometric averaging) the error probabilities of all agents vanish with time at the same exponential rate. 
Moreover, when the observations are statistically independent across the agents, this exponent is the one attained by the optimal centralized Bayesian system, i.e., by the maximum a posteriori probability (MAP) classifier~\cite{KayDet}.

The equivalence between the error exponents attained by the agents may contribute to create a ``folklore'' in decentralized decision-making, namely, that for sufficiently large amounts of data all agents share the same decision performance. Likewise, the equivalence with the optimal centralized exponent may lead to the conviction that decentralized decision-making is asymptotically optimal. 
These conclusions are also encouraged by similar findings obtained for decentralized \emph{estimation/regression} problems, where it has been shown that all agents asymptotically attain the same mean-square-error performance~\cite{chenSayedLearningBehavior2015PartI,chenSayedLearningBehavior2015PartII}.

Unfortunately, in this work we show that the asymptotic optimality of decentralized decision-making does not hold, as we now explain. To start with, observe that large deviations only deal with the error exponent. Hence, we may have two error probabilities, say $p_{1,t}=e^{-t}$ and $p_{2,t}=100 \, e^{-t}$ (where $t$ represents the time index), which, albeit sharing the same exponent, differ by two orders of magnitude~\cite{DemboZeitouni,DenHollander,MattaBordignonSayedBook}. As a result, the equivalence itself in terms of error exponents is not enough to claim the asymptotic optimality of decentralized decision-making. 
Motivated by this issue, recent works examined the higher-order terms (i.e., the terms beyond the error exponent) that affect the error probability in traditional social learning~\cite{HuangWang,ourEUSIPCO2025,ourEUSIPCOpaperarxiv2025}.

In~\cite{HuangWang}, for a binary detection problem with bounded likelihood ratios and for a binary Gaussian shift-in mean problem, the authors compute closed-form \emph{upper bounds} on the asymptotic ratio between the individual agents' error probabilities and the optimal centralized probability. That is, they show that the loss due to decentralization cannot be larger than a certain amount. 
For the Gaussian case, the \emph{exact} asymptotic ratio between the decentralized and the centralized error probabilities is instead evaluated in~\cite{ourEUSIPCOpaperarxiv2025}, revealing the analytical form of the loss experienced by traditional social learning with respect to the optimal centralized classifier. In particular, this loss is persistent in the long run, i.e., it cannot be neglected asymptotically. 

However, it is not clear whether the aforementioned gap is due to the specific decision strategy or is an intrinsic, irreducible gap due to decentralization. In this work, we answer this question.
Specifically, in Sec.~\ref{sec:decperf} we establish a fundamental limit on the performance of \emph{any decentralized decision strategy}. 
We ascertain that there exists an \emph{irreducible gap} between the performance of any decentralized strategy and the optimal centralized strategy. We characterize this loss in terms of a closed-form analytical expression that highlights the interplay between the network topology (through the distances between the agents over the network graph) and the classification problem (through the moment generating functions of the log likelihood ratios).
In Sec.~\ref{sec:tradSL}, we show that traditional social learning does not attain the fundamental limit. 
Then, in Sec.~\ref{sec:ex1} we examine how the performance scales with the network size under typical network topologies, whereas in Sec.~\ref{sec:ex2} we show how the interplay between the topology and the agents' informativeness affects the performance loss.

\section{Basic Quantities and Notation}
In this section we collect the main elements describing the inference problem and the network structure.

\subsection{Inference Problem}
We consider a decision-making process where a network of agents is interested in learning the state of a given phenomenon of interest, which is represented by a \emph{hypothesis} $\theta$ belonging to a discrete finite set $\Theta$. The hypotheses are distributed according to a \emph{prior} probability mass function (pmf) denoted by $\pi$. To avoid trivial cases, we assume that all entries of $\pi$ are nonzero.

The ensemble of network agents is represented by the set $\mathcal{K}\triangleq\{1,2,\ldots,K\}$.
At each time instant $t\in\mathbb{N}$, each agent $k\in\mathcal{K}$ gathers an observation $\bm{x}_{k,t}\in\mathbb{R}^{n_k}$ (we use bold notation for random quantities), where the space dimensionalities $n_k$ can be heterogeneous across the agents. The data are assumed independent over time and space (i.e., across the agents).

The generative mechanism for the data is as follows. Given the true hypothesis (denoted by $\thetatrue$), the observations for agent $k$ are identically distributed over time, according to some \emph{likelihood} model $\ell_k(x_k|\thetatrue)$ that, as a function of $x_k$, is a probability density function (pdf) on $\mathbb{R}^{n_k}$. For every fixed agent $k\in\mathcal{K}$, the pdfs $\ell_k(x_k|\theta)$ have the same support for all $\theta\in\Theta$. We define the log likelihood ratio (LLR) of agent $k$ at time $t$ as
\begin{equation}
\bm{\lambda}_{k,t}(\thetatrue,\theta)\triangleq\ln\frac{\ell_k(\bm{x}_{k,t}|\thetatrue)}{\ell_k(\bm{x}_{k,t}|\theta)},\qquad \theta\neq\thetatrue.
\label{eq:definizioneprimigenia}
\end{equation}

\begin{assumption}[\textbf{Likelihoods}]
\label{assum:luckylike}
For each pair of hypotheses $(\thetatrue,\theta)$ with $\theta\neq\thetatrue$, and each agent $k\in\mathcal{K}$, we consider the following two cases: $i)$ if the likelihood models $\ell_k(\bm{x}_k|\thetatrue)$ and $\ell_k(\bm{x}_k|\theta)$ are equal, then agent $k$ is unable to distinguish these hypotheses and the log likelihood ratio $\bm{\lambda}_{k,t}(\thetatrue,\theta)$ is equal to zero (with probability $1$); $ii)$ otherwise, the log likelihood ratio $\bm{\lambda}_{k,t}(\thetatrue,\theta)$ is a random variable distributed, conditionally on $\thetatrue$, according to some pdf with positive mean
\begin{equation}
\mathbb{E}_{\thetatrue}\,\bm{\lambda}_{k,t}(\thetatrue,\theta)\triangleq D_k(\thetatrue||\theta)>0.
\end{equation}
Here, the subscript in the expectation operator $\mathbb{E}_{\thetatrue}$ denotes that the expectation is computed under the true hypothesis $\thetatrue$, whereas  $D_k(\thetatrue||\theta)$ denotes the Kullback-Leibler (KL) divergence between $\ell_k(\bm{x}_k|\thetatrue)$ and $\ell_k(\bm{x}_k|\theta)$~\cite{CoverThomas}.\hfill$\square$
\end{assumption}
We rule out the extreme situation where the decision problem cannot be solved by the network of agents. In other words, the decision problem is \emph{identifiable}, which corresponds to saying that, for each pair of hypotheses, the KL divergence between the \emph{joint} likelihoods across the agents is positive. In view of the additivity of the KL divergence over independent observations, the identifiability condition translates into the following assumption: 
\begin{assumption}[\textbf{Identifiability}]
\label{assum:ident}
The decision problem is globally identifiable by the network of agents, which means that, for all $\thetatrue, \theta \in\Theta$, with $\theta\neq\thetatrue$,
\begin{equation}
\sum_{k=1}^K D_k(\thetatrue||\theta)=
\sum_{j\in\mathcal{J}(\thetatrue,\theta)} D_j(\thetatrue||\theta)
>0,
\end{equation}
where $\mathcal{J}(\thetatrue,\theta)\neq\emptyset$ is the set of agents for which the pair $(\thetatrue,\theta)$ is identifiable.\hfill$\square$
\end{assumption}
In our analysis, a critical role will be played by the log moment generating function (LMGF) of the log likelihood ratios $\bm{\lambda}_{k,t}(\thetatrue,\theta)$, which we denote by
\begin{equation}
\Lambda_k(s;\thetatrue,\theta)
\triangleq \ln\mathbb{E}_{\thetatrue}\exp\Big(
s \, \bm{\lambda}_{k,t}(\thetatrue,\theta)
\Big),
\label{eq:LMGFdef}
\end{equation}
where we see that the LMGF does \emph{not} depend on $t$ due to the identical distribution over time. We assume the following regularity condition.
\begin{assumption}[\textbf{LMGF of the log likelihood ratios}]
\label{assum:finiteMGFs}
For all $k\in\mathcal{K}$, and all $\thetatrue, \theta\in\Theta$, with $\theta\neq\thetatrue$, we have
\begin{equation}
\Lambda_k(s;\thetatrue,\theta)<\infty\quad \forall s\in\mathbb{R}.
\end{equation}
\hfill$\square$
\end{assumption}
We also introduce the \emph{global} LMGF corresponding to the sum variable
\begin{equation}
\bm{\lambda}_t(\thetatrue,\theta)\triangleq \sum_{k=1}^K \bm{\lambda}_{k,t}(\thetatrue,\theta)
=
\sum_{j\in\mathcal{J}(\thetatrue,\theta)} \bm{\lambda}_{j,t}(\thetatrue,\theta),
\label{eq:sumLLRvariable}
\end{equation}
namely,
\begin{align}
\Lambda(s;\thetatrue,\theta)
&\triangleq 
\ln\mathbb{E}_{\thetatrue}
\exp
\Big(s\,\bm{\lambda}_{t}(\thetatrue,\theta)\Big)
\nonumber\\
&=
\sum_{k=1}^K \Lambda_k(s;\thetatrue,\theta)
=
\sum_{j\in\mathcal{J}(\thetatrue,\theta)} \Lambda_j(s;\thetatrue,\theta),
\label{eq:sumLMGF}
\end{align}
where the second-to-last equality follows from the additivity of the LMGF for independent random variables, whereas the last equality holds because $\bm{\lambda}_{j,t}(\thetatrue,\theta)=0$ (almost surely) when $j\notin\mathcal{J}(\thetatrue,\theta)$, i.e., when $\thetatrue$ and $\theta$ are not identifiable for agent $j$.

\subsection{Network Model and Decentralization}
We assume that the agents can communicate according to a certain network topology, which is described by a graph, whose vertices represent the agents, and whose \emph{directed} edges represent a communication link from one agent to another. If an edge exists, originating from agent $j$ and pointing toward agent $k$, then we say that $k$ can receive information from $j$. Moreover, each agent can always use its own information, which means there is a self-loop connecting each agent to itself. We define the \emph{directed} neighborhood $\mathcal{N}_k$ for agent $k$ as the set of agents from which $k$ can receive information through a directed link (i.e., in a single hop), including $k$ itself.

\begin{definition}[\textbf{Network distances}]
\label{def:dist}
For each pair $(j,k)$, we denote by $d_{jk}$ the (directed) distance that must be traversed, over the graph, to reach node $k$ starting from node $j$. More precisely, we say that $d_{jk}=m$ when the shortest path from $j$ to $k$ is equal to $m+1$ hops.\footnote{Typically, in the graph literature, when the number of hops is equal to $m+1$, the distance is defined as $m+1$. Instead, we define the distance by subtracting $1$ since it is more convenient, in our analysis, to consider that neighboring agents lie at distance zero.} This means that  $d_{jk}=0$ when $j\in\mathcal{N}_k$.\hfill$\square$
\end{definition}
\begin{assumption}[\textbf{Connectivity}]
\label{assum:net}
We assume $d_{jk}<\infty$ for all pairs $(j,k)$, i.e., any agent can be reached through a path that originates at any other agent.\hfill$\square$ 
\end{assumption}

\begin{figure}
\centering
\includegraphics[width=0.5\linewidth]{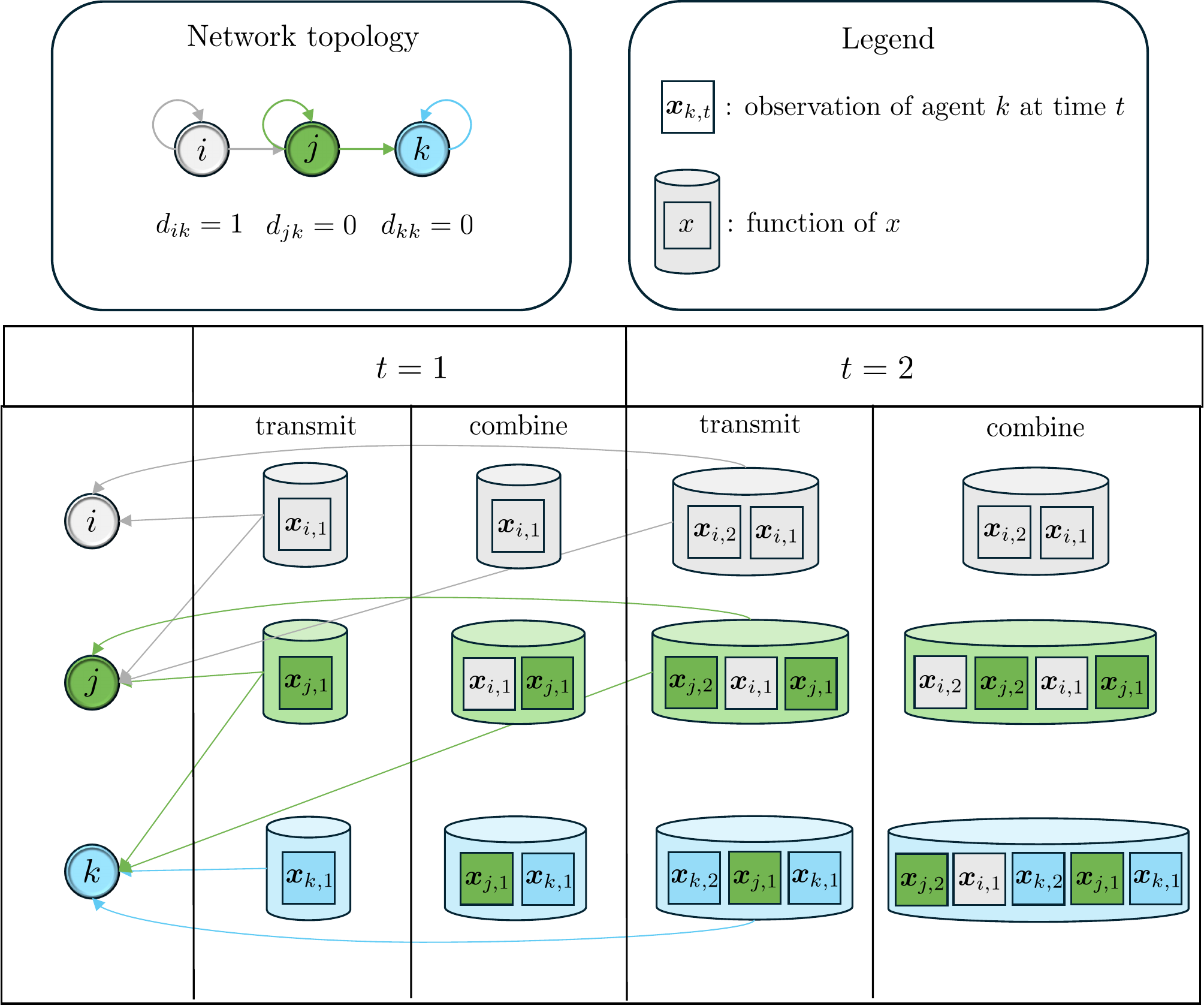}
\caption{Pictorial illustration of the flow of information across a decentralized system. As specified in the legend, the cylindrical container represents a generic function of the observations within it.}
\label{fig:pictorial}
\end{figure}

The next definition summarizes the essential constraints that a decentralized inference system must obey.

\begin{definition}[\textbf{Decentralized inference system}]
\label{def:idec}
By a decentralized system, we mean a network of agents that learn through successive iteration steps $t\in\mathbb{N}$. 
At each iteration, all agents collect a fresh observation. 
Then, they compute and deliver over the network a function (e.g., a belief or a likelihood ratio) that depends on the fresh observation and on the information received from the neighbors until time $t$. Then, this decentralized process obeys the following constraints: 
\begin{itemize}
\item Agent $k$ can receive information only from its (directed) neighbors $j\in\mathcal{N}_k$, namely, from the agents for which $d_{jk}=0$.
\item 
At time $t$, agent $k$ can compute a decision statistic that can be formally written as
\begin{equation}
f_{k,t}\Big(
\big\{\left\{
\bm{x}_{j,\tau}
\right\}_{\tau=1}^{t-d_{jk}}
\big\}_{j=1}^K
\Big),    
\end{equation}
for some function $f_{k,t}$ (which can also be a vector-valued function, such as a probability vector). Here, $\{\{
\bm{x}_{j,\tau}
\}_{\tau=1}^{t-d_{jk}}
\}_{j=1}^K$ represents the largest ensemble of data that can reach agent $k$ at time $t$.\hfill$\square$ 
\end{itemize}
\end{definition}

Observe that, regardless of the particular learning algorithm, after $t$ steps the information received by agent $k$ cannot be based on observations that are too ``recent`` and originate from agents that are too ``distant`` from $k$.
This concept can be more conveniently illustrated by means of  Fig.~\ref{fig:pictorial}. Consider the three agents, $i$, $j$, and $k$, connected as shown in the upper-left panel in the figure. 
At time $t=1$, agent $k$ collects a new local observation $\bm{x}_{k,1}$ and receives information from its neighbor $j$. This information is, in general, some function of the data $\bm{x}_{j,1}$. Similarly, agent $j$ observes $\bm{x}_{j,1}$ and receives some function of the data $\bm{x}_{i,1}$. 
The agents can then combine the information received from their neighbors. For example, agent $k$ can compute a certain function of $\bm{x}_{j,1}$ and $\bm{x}_{k,1}$. Likewise, agent $j$ will compute a certain function of $\bm{x}_{i,1}$ and $\bm{x}_{j,1}$.

Next, at $t=2$, agent $k$ (along with the fresh observation $\bm{x}_{k,2}$), will receive some function of $\bm{x}_{i,1}$ and $\bm{x}_{j,1}$ from agent $j$. After the combination phase, agent $k$ will compute a certain decision statistic from the received (functions of) data. Formally, this decision statistic can be represented as some global function of $\bm{x}_{j,2}, \bm{x}_{k,2}, \bm{x}_{i,1}, \bm{x}_{j,1},$ and $\bm{x}_{k,1}$.

Note that, at $t=2$, agent $k$ received no information related to the \emph{current} observation $\bm{x}_{i,2}$, since it takes $2$ hops to reach $k$ from $i$. In contrast, agent $k$ has already received information related to the \emph{previous-lag} observation $\bm{x}_{i,1}$, because at $t=1$ (first hop) agent $j$ received information from agent $i$, and at $t=2$ agent $j$ forwarded to agent $k$ a function that also incorporates information related to $\bm{x}_{i,1}$.

\section{Decentralizated Performance}
\label{sec:decperf}

We are interested in establishing the best attainable performance (i.e., the lowest error probability) within the class of decentralized strategies identified by Definition~\ref{def:idec}. 
To this end, we observe preliminarily that any function (i.e., transformation) of the raw data cannot bring more information than the raw data themselves. Therefore, it is understood that the \emph{maximum possible amount of information available to agent $k$ at time $t$} is given by the ensemble of raw data
\begin{equation}
\left\{\left\{
\bm{x}_{j,\tau}
\right\}_{\tau=1}^{t-d_{jk}}
\right\}_{j=1}^K.
\label{eq:idealamount}
\end{equation}
This is the largest ensemble of data that can reach agent $k$ at time $t$, compatibly with the decentralization constraints in Definition~\ref{def:idec}. Note that this is not the total amount of data available over the graph until the time instant $t$. This is because some agents may be further away from $k$.

Then, the \emph{best decision strategy} (i.e., the strategy minimizing the error probability) for agent $k$ is the maximum a posteriori probability (MAP) rule applied to the ensemble of data \eqref{eq:idealamount}. 
Any decentralized strategy belonging to the family described by Definition~\ref{def:idec} will feature an error probability that cannot be lower than the error probability achieved by the MAP classifier operating on the data set \eqref{eq:idealamount}. 

We will now examine the MAP classifier. In order to implement it, we need to construct the posterior belief of $\theta$ given the data \eqref{eq:idealamount}, which is given by
\begin{equation}
\bm{\mu}_{k,t}(\theta)\propto
\pi(\theta)\,
\prod_{j=1}^K
\prod_{\tau=1}^{t-d_{jk}}
\ell_j(\bm{x}_{j,\tau}|\theta),
\label{eq:beliefs}
\end{equation}
where, as usual in Bayesian theory, the proportionality sign $\propto$ hides the normalization factor represented by the marginal pdf of the data.
According to \eqref{eq:beliefs}, the log belief ratio of agent $k$ at time $t$ is
\begin{align}
\bm{\beta}_{k,t}(\thetatrue,\theta)
&\triangleq\ln\frac{\bm{\mu}_{k,t}(\thetatrue)}{\bm{\mu}_{k,t}(\theta)}\nonumber\\
&=
\ln\frac{\pi(\thetatrue)}{\pi(\theta)} + \sum_{j=1}^K\sum_{\tau=1}^{t-d_{jk}} \bm{\lambda}_{j,\tau}(\thetatrue,\theta).
\label{eq:betaktinitdef}
\end{align}
In order to investigate the loss arising from the decentralization, we need to consider as a benchmark the centralized MAP classifier that is based on \emph{all data}. 
The log belief ratio for this optimal system is obtained by setting $d_{jk}=0$ for all $j$ and $k$ in \eqref{eq:betaktinitdef}, and will be denoted by $\bm{\beta}^{\mathrm{cen}}_t(\thetatrue,\theta)$.

\subsection{Definition of some useful quantities}
Preliminarily, it is convenient to list some definitions that characterize the main quantities used in the forthcoming theorems.

\vspace*{5pt}
\noindent
\emph{Exact asymptotics.} 
For two sequences $a_t$ and $b_t$, the notation $a_t\simeq b_t$ signifies that
\begin{equation}
\lim_{t\rightarrow\infty} \frac{a_t}{b_t}=1.
\label{eq:exasydef}
\end{equation}
Following a standard terminology used in statistics~\cite{DemboZeitouni}, specifically in the framework of \emph{exact asymptotics}, when \eqref{eq:exasydef} holds, we say that $b_t$ is an exact asymptotic approximation for $a_t$ (or vice versa).\footnote{The qualification ``exact'' is used in the literature to indicate that other forms of asymptotic approximations, such as \emph{large deviations}, do not guarantee the asymptotic equivalence in \eqref{eq:exasydef}, but only equivalence at the leading exponential order; see also the discussion in the introduction.}

\vspace*{5pt}
\noindent
\emph{Error probabilities.}
For the optimal decentralized system, we define the binary or pairwise error probability for agent $k$ at time $t$, relative to choosing some wrong hypothesis $\theta\neq\thetatrue$ (the subscript to the probability operator indicates that the probability is evaluated under the true hypothesis $\thetatrue$):
\begin{equation}
p_{k,t}(\thetatrue,\theta)\triangleq \mathbb{P}_{\thetatrue}[\bm{\beta}_{k,t}(\thetatrue,\theta)\leq 0]
\label{eq:binerprob}
\end{equation}
and the total error probability
\begin{equation}
p_{k,t}\triangleq \sum_{\thetatrue\in\Theta}\pi(\thetatrue)\,
\mathbb{P}_{\thetatrue}\left[\min_{\theta\neq\thetatrue}\bm{\beta}_{k,t}(\thetatrue,\theta)\leq 0\right].
\label{eq:totalerrdef}
\end{equation}
The corresponding quantities for the optimal centralized system will be denoted by $p_t^{\mathrm{cen}}(\thetatrue,\theta)$ and $p_t^{\mathrm{cen}}$, respectively.

\vspace*{5pt}
\noindent
\emph{Minimizer of the network LMGF.}
Consider an agent $j\in\mathcal{J}(\thetatrue,\theta)$, i.e., such that $\mathbb{E}_{\thetatrue}\,\bm{\lambda}_{j,t}(\thetatrue,\theta)>0$. By assumption, $\ell_j(x_j|\thetatrue)$ and $\ell_j(x_j|\theta)$ have the same support, and since they both integrate to $1$ (as they are pdfs), the log likelihood ratio $\bm{\lambda}_{j,t}(\thetatrue,\theta)$ takes on negative and positive values with nonzero probability. 
From the same arguments used in the discussion following Assumption~\ref{assum:csijtau} in Appendix~\ref{app:asysums}, we conclude that: the LMGF $\Lambda(s;\thetatrue,\theta)$ is strictly convex; the equation (we use the prime notation to denote differentiation with respect to the first argument)
\begin{equation}
\Lambda^{\prime}(s;\thetatrue,\theta)=0
\label{eq:stateqth}
\end{equation}
has a unique solution $s(\thetatrue,\theta)<0$; and
\begin{equation}
\Lambda\Big(s(\thetatrue,\theta);\thetatrue,\theta\Big)<0.
\end{equation} 
Furthermore, it is easily verified that $\Lambda(-1;\thetatrue,\theta)=\Lambda(0,\thetatrue,\theta)=0$, which implies that $s(\thetatrue,\theta)\in (-1,0)$.

\vspace*{5pt}
\noindent
\emph{Error exponents.}
The quantity
\begin{equation}
\mathcal{E}(\thetatrue,\theta)
\triangleq 
-\Lambda\Big(s(\thetatrue,\theta);\thetatrue,\theta\Big)>0
\label{eq:errexpth1def}
\end{equation}
will be seen in Theorem~\ref{th:theorem1maintext} to represent the error exponent governing the exponential rate of convergence to zero of the pairwise error probabilities $p_t^{\mathrm{cen}}(\thetatrue,\theta)$ and $p_{k,t}(\thetatrue,\theta)$ for each agent $k\in\mathcal{K}$. We also introduce the worst-case (i.e., the smallest) error exponent across any pair of hypotheses,
\begin{equation}
\mathcal{E}^\star\triangleq \min_{
\substack{(\thetatrue,\theta)\in\Theta^2
\\
\theta\neq\thetatrue}} \mathcal{E}(\thetatrue,\theta).
\label{eq:smallestexp}
\end{equation}
Likewise, the set
\begin{equation}
\mathcal{S}^\star\triangleq \argmin_{
\substack{(\thetatrue,\theta)\in\Theta^2
\\
\theta\neq\thetatrue}} \mathcal{E}(\thetatrue,\theta)
\label{eq:setsmallestexp}
\end{equation}
collects all the (ordered) pairs of hypotheses sharing the smallest error exponent. 

\vspace*{5pt}
\noindent
\emph{Asymptotic variance correction.}
The quantity
\begin{equation}
\sigma^2(\thetatrue,\theta)
\triangleq 
s^2(\thetatrue,\theta)\,
\Lambda^{\prime\prime}\Big(s(\thetatrue,\theta);\thetatrue,\theta\Big)
\label{eq:varianceth1def}
\end{equation}
will be seen in Theorem~\ref{th:theorem1maintext} to represent an asymptotic variance correction used to build the approximation for the pairwise error probabilities. 

\vspace*{5pt}
\noindent
\emph{Prior correction.}
The quantity
\begin{equation}
\eta_\pi(\thetatrue,\theta)\triangleq
\left(
\dfrac{\pi(\theta)}{\pi(\thetatrue)}\right)^{|s(\thetatrue,\theta)|}.
\label{eq:priortermdef}
\end{equation}
will be seen in Theorem~\ref{th:theorem1maintext} to represent a correction, accounting for uneven prior assignments, used to build the approximation for the pairwise error probabilities.

\subsection{Main results}
\begin{theorem}[\textbf{Exact asymptotics for the pairwise error probabilities}]
\label{th:theorem1maintext}
Let Assumptions~\ref{assum:luckylike}--\,\ref{assum:net} be satisfied. 
For each pair $(\thetatrue,\theta)$ with $\theta\neq\thetatrue$, let the log likelihood ratios $\bm{\lambda}_{k,t}(\thetatrue,\theta)$ fulfill, for all $k\in\mathcal{K}$ and $t\in\mathbb{N}$, Assumption~\ref{assum:csijtau} from Appendix~\ref{app:asysums}. 
Then, the function
\begin{equation}
\mathscr{P}_{t}^{\mathrm{cen}}(\thetatrue,\theta)
\triangleq
\frac{
\eta_\pi(\thetatrue,\theta)}{\sqrt{2\pi\,\sigma^2(\thetatrue,\theta)}}
\,
\frac{
e^{
-t\,\mathcal{E}(\thetatrue,\theta)
}
}
{\sqrt{t}}
\label{eq:ptcenappdef}
\end{equation}
is an exact asymptotic approximation for the pairwise error probability of the centralized system, namely, we have
\begin{equation}
p_{t}^{\mathrm{cen}}(\thetatrue,\theta)
\simeq \mathscr{P}_{t}^{\mathrm{cen}}(\thetatrue,\theta).
\label{eq:binarypt}
\end{equation}
The pairwise error probability of each agent $k\in\mathcal{K}$ in the optimal decentralized system is instead approximated as
\begin{equation}
p_{k,t}(\thetatrue,\theta)
\simeq 
\mathscr{P}_{k,t}(\thetatrue,\theta),
\label{eq:binarypkt}
\end{equation}
where
\begin{equation}
\mathscr{P}_{k,t}(\thetatrue,\theta)
\triangleq
\mathscr{P}_{t}^{\mathrm{cen}}(\thetatrue,\theta)\,
\mathscr{L}_k(\thetatrue,\theta),
\label{eq:pktappdef}
\end{equation}
and the loss term $\mathscr{L}_k(\thetatrue,\theta)$ is given by\footnote{
Given two hypotheses $\theta_1$ and $\theta_2$, from \eqref{eq:definizioneprimigenia} and \eqref{eq:LMGFdef} we see that
\begin{equation}
\Lambda_j(s;\theta_1,\theta_2)=\Lambda_j(-1-s;\theta_2,\theta_1).
\label{eq:LLRLMGFfundident}
\end{equation}
In particular, Eq. \eqref{eq:LLRLMGFfundident} implies that $s(\theta_2,\theta_1)=-1-s(\theta_1,\theta_2)$, which in turn yields $\Lambda_j(s(\theta_1,\theta_2);\theta_1,\theta_2)=\Lambda_j(s(\theta_2,\theta_1);\theta_2,\theta_1)$. As a result, the loss defined by \eqref{eq:lossth1} is symmetric with respect to its arguments, namely, $\mathscr{L}_k(\theta_1,\theta_2)=\mathscr{L}_k(\theta_2,\theta_1)$. 
}
\begin{equation}
\mathscr{L}_k(\thetatrue,\theta)\triangleq
\exp
\left(
-\sum_{j=1}^K d_{jk}\,
\Lambda_j\Big(s(\thetatrue,\theta);\thetatrue,\theta\Big)
\right)
\geq 1.
\label{eq:lossth1}
\end{equation}
\end{theorem}
\begin{IEEEproof}
The result follows from Theorem~\ref{th:theorapp} in Appendix~\ref{app:asysums}, with the following associations:
\begin{align}
&\mathcal{J}\mapsto \mathcal{J}(\thetatrue,\theta),
\quad
& &\bm{\xi}_{j,\tau}\mapsto \bm{\lambda}_{j,\tau}(\thetatrue,\theta),\label{eq:firstassoc}
\\
&\gamma\mapsto \ln\frac{\pi(\thetatrue)}{\pi(\theta)},
\quad
& &d_j\mapsto d_{jk},
\\
&\Lambda_j(s)\mapsto \Lambda_j(s;\thetatrue,\theta),
\quad
& &\Lambda(s)\mapsto \Lambda(s;\thetatrue,\theta).\label{eq:lastassoc}
\end{align}
It is straightforward to verify that, under Assumptions~\ref{assum:luckylike}--\,\ref{assum:net}, the log likelihood ratios $\bm{\lambda}_{j,\tau}(\thetatrue,\theta)$, for $j\in\mathcal{J}(\thetatrue,\theta)$, satisfy the conditions required in Assumption~\ref{assum:csijtau} for the random variables $\bm{\xi}_{j,\tau}$. 
Then, Eq. \eqref{eq:binarypkt} follows by performing the substitutions in \eqref{eq:firstassoc}--\eqref{eq:lastassoc}. The companion result for the centralized system in \eqref{eq:binarypt} is simply obtained by setting $d_{jk}=0$ for all $j$ and $k$. 

Moreover, the inequality in \eqref{eq:lossth1} holds because (recall that $\mathcal{J}(\thetatrue,\theta)\neq \emptyset$ in view of Assumption~\ref{assum:ident}) 
\begin{equation}
\Lambda_j\Big(s(\thetatrue,\theta);\thetatrue,\theta\Big)<0\quad \forall j\in\mathcal{J}(\thetatrue,\theta).
\label{eq:neqLamjnewlastfin}
\end{equation}
This relation can be explained by observing that $s(\thetatrue,\theta)\in (-1,0)$ and $\Lambda_j(s;\thetatrue,\theta)<0$ for all $s\in (-1,0)$, which holds since $\Lambda_j(-1;\thetatrue,\theta)=\Lambda_j(0,\thetatrue,\theta)=0$ and, for $j\in\mathcal{J}(\thetatrue,\theta)$,   
$\Lambda_j(s;\thetatrue,\theta)$ is a strictly convex function; see the discussion following Assumption~\ref{assum:csijtau} in Appendix~\ref{app:asysums}.
\end{IEEEproof}

\begin{theorem}[\textbf{Exact asymptotics for the total error probabilities}]
\label{th:mainfinalth}
Under the same assumptions used in Theorem~\ref{th:theorem1maintext}, we construct
the following approximations for the total error probabilities in terms of the pairwise error probability approximations from \eqref{eq:ptcenappdef} and \eqref{eq:pktappdef}:
\begin{align}
\mathscr{P}_t^{\mathrm{cen}}&\triangleq \sum_{\thetatrue\in\Theta}\pi(\thetatrue)\sum_{\theta\neq\thetatrue}\mathscr{P}_t^{\mathrm{cen}}(\thetatrue,\theta),
\label{eq:overallptappdef}
\\
\mathscr{P}_{k,t}&\triangleq \sum_{\thetatrue\in\Theta}
\pi(\thetatrue)\sum_{\theta\neq\thetatrue}
\mathscr{P}_{k,t}(\thetatrue,\theta).
\label{eq:overallpktappdef}
\end{align}
Then, it holds that
\begin{equation}
p_t^{\mathrm{cen}}\simeq
\mathscr{P}_{t}^{\mathrm{cen}},\qquad
p_{k,t}\simeq\mathscr{P}_{k,t}.
\end{equation}
\end{theorem}

\begin{IEEEproof}
We prove the claim for the decentralized system, from which the proof for the centralized system will follow immediately. 
To start with, observe that the approximation $\mathscr{P}_{k,t}$ in \eqref{eq:overallpktappdef} is a sum of the approximations $\mathscr{P}_{k,t}(\thetatrue,\theta)$ defined in \eqref{eq:pktappdef}, which decay exponentially with exponents $\mathcal{E}(\thetatrue,\theta)$. 
Accordingly, only the terms featuring the smallest error exponent $\mathcal{E}^{\star}$ in \eqref{eq:smallestexp} will dominate asymptotically. 
Thus, we are interested in evaluating the limit of the following ratio:
\begin{align}
&\frac{\mathscr{P}_{k,t}}
{
e^{-t\,\mathcal{E}^\star}/\sqrt{2\pi\,t}
}
=
\frac
{\sum_{\thetatrue\in\Theta}
\pi(\thetatrue)\sum_{\theta\neq\thetatrue} \mathscr{P}_{k,t}(\thetatrue,\theta)
}
{e^{-t\,\mathcal{E}^\star}/\sqrt{2\pi\,t}}
\nonumber\\
\nonumber\\
&=
\frac
{\sum_{\thetatrue\in\Theta}
\pi(\thetatrue)\sum_{\theta\neq\thetatrue} 
\mathscr{P}_{t}^{\mathrm{cen}}(\thetatrue,\theta)\,
\mathscr{L}_k(\thetatrue,\theta)
}
{e^{-t\,\mathcal{E}^\star}/\sqrt{2\pi\,t}}
\nonumber\\
\nonumber\\
&=
\sum_{\thetatrue\in\Theta}
\pi(\thetatrue)\sum_{\theta\neq\thetatrue}
\dfrac{
\eta_\pi(\thetatrue,\theta)}{\sigma(\thetatrue,\theta)}
\,
e^{
-t\big(\mathcal{E}(\thetatrue,\theta)-\mathcal{E}^\star\big)
}
\mathscr{L}_k(\thetatrue,\theta),
\end{align}
where, in the three equalities, we use  \eqref{eq:overallpktappdef}, \eqref{eq:pktappdef}, and  \eqref{eq:ptcenappdef}, respectively.
From the definition of the smallest exponent $\mathcal{E}^{\star}$ in \eqref{eq:smallestexp} and the set $\mathcal{S}^\star$ in \eqref{eq:setsmallestexp}, we see that
\begin{equation}
\lim_{t\rightarrow\infty}\frac{\mathscr{P}_{k,t}}
{\dfrac{e^{-t\,\mathcal{E}^\star}}{\sqrt{2\pi\,t}}}=
\sum_{(\thetatrue,\theta)\in\mathcal{S}^\star}
\pi(\thetatrue)
\dfrac{\eta_\pi(\thetatrue,\theta)}{\sigma(\thetatrue,\theta)}\,\mathscr{L}_k(\thetatrue,\theta),
\label{eq:worstcterm}
\end{equation}
We conclude from \eqref{eq:worstcterm} that, to prove that $p_{k,t}\simeq \mathscr{P}_{k,t}$, it suffices to show that 
\begin{equation}
\lim_{t\rightarrow\infty}\frac{p_{k,t}}
{\dfrac{e^{-t\,\mathcal{E}^\star}}{\sqrt{2\pi\,t}}}
=
\sum_{(\thetatrue,\theta)\in\mathcal{S}^\star}
\pi(\thetatrue)
\dfrac{\eta_\pi(\thetatrue,\theta)}{\sigma(\thetatrue,\theta)}\,\mathscr{L}_k(\thetatrue,\theta).
\label{eq:equivstatem}
\end{equation}
We proceed then to prove \eqref{eq:equivstatem}.

By applying the union bound, the total error probability is upper bounded by the sum of the pairwise error probabilities as follows:
\begin{equation}
p_{k,t}
\leq
\sum_{\thetatrue\in\Theta}\pi(\thetatrue)
\sum_{\theta\neq\thetatrue} p_{k,t}(\thetatrue,\theta),
\label{eq:simplyUB}
\end{equation}
which, in view of the subadditivity of the limit superior, implies that
\begin{equation}
\limsup_{t\rightarrow\infty}
\frac{p_{k,t}}
{\dfrac{e^{-t\,\mathcal{E}^\star}}{\sqrt{2\pi\,t}}}
\leq
\sum_{\thetatrue\in\Theta}\pi(\thetatrue)
\sum_{\theta\neq\thetatrue} 
\limsup_{t\rightarrow\infty}
\frac{p_{k,t}(\thetatrue,\theta)}
{\dfrac{e^{-t\,\mathcal{E}^\star}}{\sqrt{2\pi\,t}}}.
\label{eq:simplyUBlimsup}
\end{equation}
Now, we know from Theorem~\ref{th:theorem1maintext} that the pairwise error probabilities $p_{k,t}(\thetatrue,\theta)$ are asymptotically equivalent to the approximations $\mathscr{P}_{k,t}(\thetatrue,\theta)$ defined by \eqref{eq:pktappdef}. This implies that in the limit \eqref{eq:simplyUBlimsup}, only the terms with the smallest exponent survive, namely,
\begin{equation}
\limsup_{t\rightarrow\infty}
\frac{p_{k,t}}
{\dfrac{e^{-t\,\mathcal{E}^\star}}{\sqrt{2\pi\,t}}}
\leq
\sum_{(\thetatrue,\theta)\in\mathcal{S}^\star}
\!\!\!
\pi(\thetatrue)
\dfrac{\eta_\pi(\thetatrue,\theta)}{\sigma(\thetatrue,\theta)}\,\mathscr{L}_k(\thetatrue,\theta).
\label{eq:simplyUBlimsup2}
\end{equation}
Therefore, the claim in \eqref{eq:equivstatem} will be proved if we show that the limit inferior fulfills \eqref{eq:simplyUBlimsup2} with the reverse inequality. 

To this end, observe that we can write
\begin{align}
&\mathbb{P}_{\thetatrue}\left[\min_{\theta\neq\thetatrue}\bm{\beta}_{k,t}(\thetatrue,\theta)\leq 0\right]\nonumber\\
&=\mathbb{P}_{\thetatrue}\Bigg[\bigcup_{\theta\neq\thetatrue}
\left\{\bm{\beta}_{k,t}(\thetatrue,\theta)\leq 0\right\}\Bigg]
\leq \mathbb{P}_{\thetatrue}\left[
\bm{\beta}_{k,t}(\thetatrue,\theta)\leq 0\right]\nonumber\\
&-\frac 1 2\,
\underbrace{
\mathbb{P}_{\thetatrue}\left[
\bm{\beta}_{k,t}(\thetatrue,\theta)\leq 0,
\,
\bm{\beta}_{k,t}(\thetatrue,\theta')\leq 0
\right]}_{\triangleq b_t(\thetatrue,\theta,\theta^{\prime})},
\label{eq:Bonferrogenuine}
\end{align}
where the last step follows from the second-order Bonferroni's inequality ~\cite{GalambosBonferroniBook}. Using \eqref{eq:Bonferrogenuine} in \eqref{eq:totalerrdef} we obtain
\begin{equation}
p_{k,t}
\geq 
\sum_{\thetatrue\in\Theta}\pi(\thetatrue)\,
\sum_{\theta\neq\thetatrue} 
\bigg\{
p_{k,t}(\thetatrue,\theta) -\frac 1 2 
\sum_{\substack{\theta'\neq \thetatrue\\
\theta'\neq \theta_{\phantom{0}}}}
b_t(\thetatrue,\theta,\theta^{\prime})
\bigg\}.
\label{eq:Bonferro}
\end{equation}
We focus on the following limit inferior:
\begin{align}
&\liminf_{t\rightarrow\infty}
\frac{p_{k,t}}
{e^{-t\,\mathcal{E}^\star}/\sqrt{2\pi\,t}}
\geq
\sum_{(\thetatrue,\theta)\in\mathcal{S}^\star}
\!\!\!
\pi(\thetatrue)
\dfrac{\eta_\pi(\thetatrue,\theta)}{\sigma(\thetatrue,\theta)}\,\mathscr{L}_k(\thetatrue,\theta)
\nonumber\\
\nonumber\\
&-\frac 1 2
\sum_{\substack{
\thetatrue\in\Theta
\\
\theta^{\phantom{\prime}}:\,(\thetatrue,\theta^{\phantom{\prime}})\in\mathcal{S}^{\star}
\\
\theta':\,(\thetatrue,\theta')\in\mathcal{S}^{\star}
}}
\limsup_{t\rightarrow\infty}
\dfrac{
\pi(\thetatrue)\,b_t(\thetatrue,\theta,\theta^{\prime})}
{
e^{-t\,\mathcal{E}^\star}/\sqrt{2\pi\,t}
},
\label{eq:Bonferroretain}
\end{align}
where the last summation has been restricted to the pairs of hypotheses featuring the smallest exponent $\mathcal{E}^{\star}$ because the probability of the intersection of two events is dominated by the probabilities of the individual events, and, hence, the joint probabilities  $b_t(\thetatrue,\theta,\theta^{\prime})$ cannot exhibit an exponent smaller than $\mathcal{E}(\thetatrue,\theta)$ and $\mathcal{E}(\thetatrue,\theta^{\prime})$. 

By applying Theorem~\ref{th:jointhreshcross} from Appendix~\ref{app:propLLR}, we conclude that the limit superior appearing in \eqref{eq:Bonferroretain} is zero, and \eqref{eq:equivstatem} is proved.
\end{IEEEproof}

\begin{corollary}[\textbf{Fundamental limit on decentralized decision-making}]
\label{th:theoremloss}
Under the same assumptions used in Theorem~\ref{th:theorem1maintext}, any decentralized strategy fulfilling the constraints in Definition~\ref{def:idec} experiences at least the following performance loss with respect to the optimal centralized system:
\begin{equation}
\mathscr{L}_k\triangleq \lim_{t\rightarrow\infty}\frac{p_{k,t}}{p_t^{\mathrm{cen}}}=   \sum_{(\thetatrue,\theta)\in\mathcal{S}^\star}
\omega(\thetatrue,\theta)\,\mathscr{L}_k(\thetatrue,\theta), 
\label{eq:overalloss}
\end{equation}
where
\begin{equation}
\omega(\thetatrue,\theta)\triangleq
\frac{\pi(\thetatrue)\,\eta_\pi(\thetatrue,\theta)/\sigma(\thetatrue,\theta)}
{
\sum\limits_{(\thetatrue^\prime,\theta^\prime)\in\mathcal{S}^\star}
\pi(\thetatrue^\prime)\,\eta_\pi(\thetatrue^\prime,\theta^\prime)/\sigma(\thetatrue^\prime,\theta^\prime)
}.
\label{eq:qqweights}
\end{equation}
\end{corollary}
\begin{IEEEproof}
Applying \eqref{eq:equivstatem} to the centralized case (i.e., setting $\mathscr{L}_k(\thetatrue,\theta)=1$), we have
\begin{equation}
\lim_{t\rightarrow\infty}\frac{p_{t}^{\mathrm{cen}}}
{\dfrac{e^{-t\,\mathcal{E}^\star}}{\sqrt{2\pi\,t}}}
=
\sum_{(\thetatrue,\theta)\in\mathcal{S}^\star}
\pi(\thetatrue)
\dfrac{\eta_\pi(\thetatrue,\theta)}{\sigma(\thetatrue,\theta)}.
\label{eq:equivstatemcent}
\end{equation}
Combining \eqref{eq:equivstatem} and \eqref{eq:equivstatemcent}, we obtain the claim of the theorem.
\end{IEEEproof}
Corollary~\ref{th:theoremloss} establishes that there exists an \emph{irreducible gap} $\mathscr{L}_k$ between the ideal decentralized system and the optimal centralized system. 
Since the performance of the ideal system represents a fundamental limit that any decentralized strategy cannot surpass, this result also implies that \emph{any decentralized strategy cannot reach the performance of the centralized strategy.}

The loss in terms of error probability can also be translated into a \emph{learning delay}, as established by the following corollary.

\begin{corollary}[\textbf{Learning delay}]
\label{th:corolla2}
In the decentralized ideal system, the learning curve pertaining to agent $k$ (i.e., the error probability $p_{k,t}$ as a function of time) experiences an irreducible delay 
\begin{equation}
T_k = \frac{\ln \mathscr{L}_k}{\mathcal{E}^{\star}} 
\label{eq:fundelaydef}
\end{equation}
with respect to the learning curve $p_t^{\mathrm{cen}}$ of the optimal centralized system. Formally, this means that
\begin{equation}
\lim_{t\rightarrow\infty} 
\frac{p_{k,t+T_k}}{p_t^{\mathrm{cen}}}=1.
\label{eq:cor2claim}
\end{equation}
\end{corollary}
\begin{IEEEproof}
From \eqref{eq:equivstatem}, we can write
\begin{equation}
\lim_{t\rightarrow\infty}\frac{p_{k,t+T_k}}
{\dfrac{e^{-(t+T_k)\,\mathcal{E}^\star}}{\sqrt{2\pi\,(t+T_k)}}}
=
\sum_{(\thetatrue,\theta)\in\mathcal{S}^\star}
\pi(\thetatrue)
\dfrac{\eta_\pi(\thetatrue,\theta)}{\sigma(\thetatrue,\theta)}\,\mathscr{L}_k(\thetatrue,\theta). 
\label{eq:delayedvers}
\end{equation}
Combining \eqref{eq:equivstatemcent} and \eqref{eq:delayedvers}, we obtain
\begin{align}
\lim_{t\rightarrow\infty}\frac{p_{k,t+T_k}}
{p_t^{\mathrm{cen}}}
&=
e^{-T_k\,\mathcal{E}^{\star}}\,
\frac{\pi(\thetatrue)\,\eta_\pi(\thetatrue,\theta)/\sigma(\thetatrue,\theta)}
{
\sum\limits_{(\thetatrue^\prime,\theta^\prime)\in\mathcal{S}^\star}
\pi(\thetatrue^\prime)\,\eta_\pi(\thetatrue^\prime,\theta^\prime)/\sigma(\thetatrue^\prime,\theta^\prime)
}\nonumber\\
&=
e^{-T_k\,\mathcal{E}^{\star}}\,\mathscr{L}_k,
\label{eq:limratiotoprovedelay}
\end{align}
where in the last step we applied the definition of $\mathscr{L}_k$ from \eqref{eq:overalloss}. The proof is completed by substituting \eqref{eq:fundelaydef} into \eqref{eq:limratiotoprovedelay}.
\end{IEEEproof}
Corollary~\ref{th:corolla2} provides further insight into the learning behavior of decentralized decision-making systems. The takeaway is that \emph{learning is irremediably slowed down due to decentralization.}
To avoid misunderstanding, we remark  that this does not imply in any manner that the decentralized system can attain the centralized performance if one accepts to wait for $T_k$ more time instants. The opposite conclusion is true! In fact, in these additional rounds, the centralized system also collects additional information and learns from it. As a result, the decentralized system will experience a persistent delay and will never attain the performance of the centralized system. 

Examining \eqref{eq:overalloss}, we see that the general expression for the learning delay $T_k$ in \eqref{eq:fundelaydef} depends in a nontrivial manner on the losses $\mathscr{L}_k(\thetatrue,\theta)$ associated with the individual pairs $(\thetatrue,\theta)\in\mathcal{S}^{\star}$ and the relative weights $\omega(\thetatrue,\theta)$ defined by \eqref{eq:qqweights}. The expression simplifies when the losses $\mathscr{L}_k(\thetatrue,\theta)$ pertaining to the pairs $(\thetatrue,\theta)$ within the set $\mathcal{S}^{\star}$ assume the same value (which includes the situation where $\mathcal{S}^{\star}$ contains a single element). Under the equal-loss condition, using \eqref{eq:lossth1} in \eqref{eq:fundelaydef}, we obtain
\begin{equation}
T_k=
\frac{\ln\mathscr{L}_k(\thetatrue,\theta)}{\mathcal{E}(\thetatrue,\theta)}
=\frac
{\sum_{j=1}^K d_{jk}\,
\Lambda_j\Big(s(\thetatrue,\theta);\thetatrue,\theta\Big)}
{\Lambda\Big(s(\thetatrue,\theta);\thetatrue,\theta\Big)},
\label{eq:simplifiedcase}
\end{equation}
where in the last step we used the definition of $\mathcal{E}(\thetatrue,\theta)$ from \eqref{eq:errexpth1def}. 
By further exploiting the definition of $\Lambda\Big(s(\thetatrue,\theta);\thetatrue,\theta\Big)$ from \eqref{eq:sumLMGF}, we see that the weights
\begin{equation}
q_j\triangleq \frac
{\Lambda_j\Big(s(\thetatrue,\theta);\thetatrue,\theta\Big)}
{\Lambda\Big(s(\thetatrue,\theta);\thetatrue,\theta\Big)}
\label{eq:simplifiedweights}
\end{equation}
are convex, i.e., they are nonnegative and add up to $1$. Accordingly, Eq. \eqref{eq:simplifiedcase} can be recast in the form
\begin{equation}
T_k=\sum_{j=1}^K q_j\, d_{jk}.
\label{eq:simplifiedelayform}
\end{equation}
This shows that, under the equal-loss condition, the learning delay experienced by agent $k$ is a convex combination of the distances $d_{jk}$ over the graph. The weight $q_j$ in \eqref{eq:simplifiedweights} accounts for the relative informativeness of agent $j$. Accordingly, we see from \eqref{eq:simplifiedelayform} that each distance $d_{jk}$ contributes to the learning delay $T_k$ depending on the ``inferential relevance'' of agent $j$.

\subsection{Achievability}
\label{sec:achiev}
In the previous section, we established the best attainable performance for any agent $k$ at any time $t$. To evaluate the performance, we computed the error probability corresponding to the maximum amount of data \eqref{eq:idealamount} that agent $k$ could access at time $t$, according to the decentralization constraints as per Definition~\ref{def:idec}. 
A communication protocol that allows each agent $k$ to access the data set \eqref{eq:idealamount} is as follows: at each time $t$, each agent forwards to its neighbors its own fresh observation and the observations received from its neighbors at time $t-1$. In these transmission steps, the agents attach to the data a label denoting which agent the relayed observation belongs to. 

The described scheme has the following drawbacks. First, to implement the MAP rule, each agent needs to know the likelihood models of the other agents, an assumption usually considered too strong in many decentralized settings; second, transmitting the raw data usually violates several constraints arising in practical applications, such as privacy or computational constraints (since the data dimensionality could be large). 
However, it is possible to overcome these limitations by means of the following strategy: at each iteration $t$, each agent forwards to its neighbors the likelihoods computed for the new observation and the likelihoods received from its neighbors at time $t-1$. We see that, in this manner, an agent need not know the likelihood models of the other agents and no raw data are transmitted (note also that the likelihood is always a scalar, so the data dimensionality is not an issue).\footnote{Technically, it is not necessary that the agents send the likelihoods for all $\theta\in \Theta$. It is possible to show that~\cite[Thm. 6.1]{MattaBordignonSayedBook} any belief function can be constructed starting from the likelihood \emph{ratios} computed with respect to some arbitrarily chosen pivot hypothesis, which requires then to share $|\Theta|-1$ likelihood ratios.}
We have thus shown that the limit performance evaluated in Sec.~\ref{sec:decperf} can be attained by a decentralized protocol where each agent possesses only local private likelihood models and where no raw data are transmitted. 

\section{Traditional Social Learning}
\label{sec:tradSL}
In Sec.~\ref{sec:achiev}, we identified two strategies that are able to attain the best possible decentralized performance. 
Observe that these strategies assume that each agent keeps trace of the data or likelihoods received from all the other agents. 
In the theory of social learning, a different approach is usually proposed, which prescribes that each agent builds its belief $\bm{\mu}_{k,t}(\theta)$ through the following strategy:
\begin{align}
\bm{\psi}_{k,t}(\theta)&=
\dfrac
{\bm{\mu}_{k,t-1}(\theta)\,\ell_{k}(\bm{x}_{k,t}|\theta)}{\sum_{\theta'\in\Theta}\bm{\mu}_{k,t-1}(\theta')\,\ell_{k}(\bm{x}_{k,t}|\theta')},
\label{eq:tradSLBayesup}
\\
\bm{\mu}_{k,t}(\theta)&=\dfrac
{\prod_{j\in\mathcal{N}_k}[\bm{\psi}_{j,t}(\theta)]^{a_{jk}}}
{\sum_{\theta'\in\Theta}\prod_{j\in\mathcal{N}_k}[\bm{\psi}_{j,t}(\theta')]^{a_{jk}}}.
\label{eq:tradSLpool}
\end{align}
We see that, at each iteration $t$, each agent $k$ performs the following steps: $i)$ first, it updates an intermediate belief $\bm{\psi}_{k,t}(\theta)$ by incorporating the local fresh information $\bm{x}_{k,t}$ by means of a local Bayesian update, based on the local likelihood $\ell_k(\bm{x}_{k,t}|\theta)$ and the previous belief $\bm{\mu}_{k,t-1}(\theta)$; $ii)$ then, it shares the intermediate beliefs over the network; $iii)$ finally, it combines the intermediate beliefs received from its neighbors by means of a weighted geometric average using some combination weights $a_{jk}\geq 0$. The weights follow the underlying network structure, in the sense that $a_{jk}=0$ when there is no directed link from $j$ to $k$. 
It is customary to collect these weights into a weighted adjacency matrix $A=[a_{jk}]$, a.k.a. combination matrix. 
Moreover, the weights can be chosen so as to make $A$ a \emph{doubly stochastic} matrix, i.e., a matrix whose entries add up to $1$ across the columns and across the rows.
It has been proved that, with doubly stochastic combination matrices, the strategy \eqref{eq:tradSLBayesup}--\eqref{eq:tradSLpool} attains the same error exponent as the centralized system~\cite{MattaBordignonSayedBook}. 

Traditional social learning is fundamentally different from the schemes attaining the limit decentralized performance that we described in Sec.~\ref{sec:achiev}. 
In traditional social learning, each agent does not keep trace of the information flowing over the network, and updates (i.e., overwrites) the same belief function continually. This process is more parsimonious in terms of communication/memory resources. However, the beliefs received by an agent usually contain common pieces of information arriving from multiple paths. Due to the continuous belief overwriting, these redundancies are not disentangled. For this reason, we do not expect traditional social learning to attain the limit performance of decentralized systems. 

To elaborate on this issue, let us consider a decision problem with two hypotheses $\theta\in\Theta=\{0,\nu\}$, which are assumed to be equally distributed \emph{a priori}, i.e., $\pi(0)=1/2$. Each agent $k$ at time $t$ collects an observation $\bm{x}_{k,t}$ distributed as a unit-variance Gaussian, with means $0$ and $\nu$ under the two hypotheses, respectively. 
Let
\begin{equation}
\varepsilon_k\triangleq \frac{\Delta}{4} \sum_{j=1}^K \sum_{\tau=1}^\infty 
\Big(
[A^\tau]_{jk}\,K - 1
\Big)^2,
\label{eq:tradSLosstermeps}
\end{equation}
where $[A^\tau]_{jk}$ denotes the $(j,k)$ entry of the power matrix $A^\tau$. 
It has been shown in~\cite{ourEUSIPCOpaperarxiv2025} that, when $A$ is a primitive matrix, the loss experienced by agent $k$ in traditional social learning \eqref{eq:tradSLBayesup}--\eqref{eq:tradSLpool} is 
\begin{equation}
\mathscr{L}_k^{\mathrm{SL}}=\exp(\varepsilon_k).
\label{eq:tradSLoss}
\end{equation}
Let us also evaluate the optimal decentralized performance loss $\mathscr{L}_k$ predicted by Theorem~\ref{th:mainfinalth}. 
It is a straightforward exercise to prove that the log likelihood ratios under the two hypotheses are distributed as follows: 
\begin{align}
\bm{\lambda}_{k,t}(0,\nu)&\sim \mathcal{G}(\Delta,2\Delta),\qquad\textnormal{under $\thetatrue=0$}
\\
\bm{\lambda}_{k,t}(\nu,0)&\sim \mathcal{G}(\Delta,2\Delta),\qquad\textnormal{under $\thetatrue=\nu$},
\end{align}
where $\Delta\triangleq \nu^2/2$ is the KL divergence between the distributions (which, in this particular example, is the same under both hypotheses), while the notation $\mathcal{G}(a,b)$ denotes a Gaussian distribution with mean $a$ and variance $b$. 
Recalling that the LMGF of a Gaussian variable with mean $a$ and variance $b$ is equal to $a t + b \, t^2 /2$, we conclude that the LMGFs of the log likelihood ratios for each agent $k$ are equal under both hypotheses and are given by  
\begin{equation}
\Lambda_k(s;0,\nu)=\Lambda_k(s;\nu,0)=\Delta \, s \left(
1+s 
\right).
\label{eq:individualLMGFex1}
\end{equation}
Likewise, the \emph{global} LMGFs in \eqref{eq:sumLMGF} are given by
\begin{equation}
\Lambda(s;0,\nu)=\Lambda(s;\nu,0)=K \Delta \, s \left(
1+s 
\right),
\label{eq:globalLMGFex1}
\end{equation}
which implies that the stationary equation \eqref{eq:stateqth} has the solution
\begin{equation}
s(0,\nu)=s(\nu,0)=-\frac 1 2.
\label{eq:s0solutonehalf}
\end{equation}
Substituting this result into \eqref{eq:individualLMGFex1}, we find that the performance losses under both hypotheses are
\begin{equation}
\mathscr{L}_k(0,\nu)=\mathscr{L}_k(\nu,0)=\exp\left(
\frac{\Delta}{4} \sum_{j=1}^K d_{jk}
\right).
\end{equation}
Since the weights in \eqref{eq:qqweights} add up to $1$, we conclude that the overall loss in \eqref{eq:overalloss} for agent $k$ is
\begin{equation}
\mathscr{L}_k=\exp\left(
\frac{\Delta}{4} \, \sum_{j=1}^K d_{jk}
\right).
\label{eq:lossex1init}
\end{equation}
We now show that the optimal loss in \eqref{eq:lossex1init} is always smaller than the loss in \eqref{eq:tradSLoss} that characterizes traditional SL. 
To this aim, we start by observing that the entries of the matrix powers, $[A^\tau]_{jk}$, are nonzero iff there exists a path of length $\tau$ that starts at $j$ and ends at $k$~\cite{HornJohnson,MattaBordignonSayedBook}. According to Definition~\ref{def:dist}, the shortest path from $j$ to $k$ has length $d_{jk}+1$, which means that $\forall \tau\leq d_{jk}$ we have $[A^\tau]_{jk}=0$, implying
\begin{align}
&\sum_{\tau=1}^{\infty} \Big(
[A^\tau]_{jk}\,K - 1
\Big)^2=
\sum_{\tau=1}^{d_{jk}} \Big(
\underbrace{[A^\tau]_{jk}\,K}_{=0} - 1
\Big)^2\nonumber\\
&+
\sum_{\tau=d_{jk}+1}^{\infty} \Big(
[A^\tau]_{jk}\,K - 1
\Big)^2\nonumber\\
&=d_{jk} +
\sum_{\tau=d_{jk}+1}^{\infty} \Big(
[A^\tau]_{jk}\,K - 1
\Big)^2.
\label{eq:secondsumspecialbound}
\end{align}
It is well-known that, when $A$ is primitive, all columns in the matrix power $A^{\tau}$ converge to the Perron vector~\cite{HornJohnson,MattaBordignonSayedBook}. For doubly stochastic matrices, this vector has all its entries equal to $1/K$~\cite{HornJohnson,MattaBordignonSayedBook}. Since, in general, the convergence of $[A^\tau]_{jk}$ to $1/K$, takes place asymptotically,\footnote{
The reader might wonder whether the matrix power entries can converge to the Perron vector entries in a finite number of steps. From a theoretical standpoint, this can happen for some exceptional topologies. 
We now show that, even in this special case, the decentralized system experiences a loss. Assume the extreme situation where the convergence is met after only two steps, i.e., $[A^{\tau}]_{jk}=1/K$ for all $\tau>1$. Note that, since all agents are connected in two hops, we have $d_{jk}=1$ for $j\notin\mathcal{N}_k$. From \eqref{eq:tradSLosstermeps}, \eqref{eq:tradSLoss}, and \eqref{eq:lossex1init}, we conclude that
\begin{equation}
\mathscr{L}_k^{\mathrm{SL}} = 
\mathscr{L}_k\times 
\exp\left(
\sum_{j\in\mathcal{N}_k} 
\Big(
a_{jk}\,K - 1
\Big)^2
\right).
\end{equation}
Thus, to achieve $\mathscr{L}_k^{\mathrm{SL}} = 
\mathscr{L}_k$, we must have $a_{jk}=1/K$ for all $j\in\mathcal{N}_k$. However, this means that agent $k$ should be connected to all other agents. Moreover, this cannot happen for all agents, but for the  trivial case where the network is fully connected.} we conclude that the summation on the RHS is positive, yielding, in view of \eqref{eq:tradSLosstermeps}, \eqref{eq:tradSLoss}, and \eqref{eq:lossex1init},
\begin{equation}
\mathscr{L}_k^{\mathrm{SL}} > \mathscr{L}_k.
\label{eq:tradSLvsdecloss}
\end{equation}
In summary, the limit decentralized performance in Theorem~\ref{th:mainfinalth} can be attained by the approaches illustrated in Sec.~\ref{sec:achiev}.
Now we have shown, by means of a counterexample, that in traditional social learning the limit decentralized performance is not attained in general. 

These results pave the way for a new research question, which can be regarded as a \emph{rate/distortion} formulation of decentralized decision-making. 
At one extreme, we can consider the scheme that attains the limit decentralized performance as the most expensive scheme in terms of communication/memory resources. 
In comparison, traditional social learning schemes save resources by maintaining and overwriting one belief vector for each agent. 
Then, after defining a suitable metric to quantify the communication/memory expense (i.e., the \emph{rate}), it would be legitimate to investigate the possibility of designing new decentralized strategies that attain an error probability (i.e., the distortion) closer to the optimal decentralized scheme, possibly at the price of an increased rate with respect to traditional SL. The most ambitious and challenging goal would be to establish the optimal rate/distortion curve for decentralized decision-making, i.e., the best error probability achievable for any rate budget.

\section{Illustrative Examples}
In this section, we apply the performance characterization provided by Theorem~\ref{th:mainfinalth} to some exemplary scenarios. Specifically, in Sec.~\ref{sec:ex1} we examine the impact of two classic graph structures on the decentralized performance loss, whereas in Sec.~\ref{sec:ex2} we examine the interplay between the inference power of the agents and their connectivity.

\subsection{Scaling Laws under Different Graphs}
\label{sec:ex1}
We want to establish how the performance slowdown scales with the number of agents, under two different network constructions, namely, a \emph{ring} topology and an \emph{Erd\H{o}s-R\'enyi random graph}. To this end, we examine the scaling law for the learning delay $T_k$ from Corollary~\ref{th:corolla2}.
In order to magnify the role of the topology itself, we focus on an inference problem where all agents are equally informed. Later, in the next section, we consider the case where different agents, with different connectivity, have unequal inference power. Specifically, in this section we consider the binary decision problem illustrated in Sec.~\ref{sec:tradSL}, for which the optimal decentralized loss is represented by \eqref{eq:lossex1init}. 
Since, in this example, the losses are equal under both hypotheses, we can use the simplified expression \eqref{eq:simplifiedelayform} to evaluate the learning delay $T_k$. 
From \eqref{eq:individualLMGFex1},  \eqref{eq:globalLMGFex1}, and \eqref{eq:s0solutonehalf}, we see that, under the considered setting, the weights in \eqref{eq:simplifiedweights} are uniform, i.e., $q_j=1/K$ for all $j\in\mathcal{K}$, yielding
\begin{equation}
T_k=\frac 1 K \sum_{j=1}^K d_{jk}.
\label{eq:arithmeticdelay}
\end{equation}
We conclude that the weighted structure in \eqref{eq:simplifiedelayform} (which, we recall, holds because the losses are equal under both hypotheses) translates into the arithmetic average of the distances from any agent $j$ to agent $k$. This is because we are considering equally informed agents and, hence, the weights $q_j$ that quantify the relative importance of each agents are uniform, therefore bringing no difference across the agents.

\begin{figure}
\centering
\includegraphics[width=0.3\linewidth]{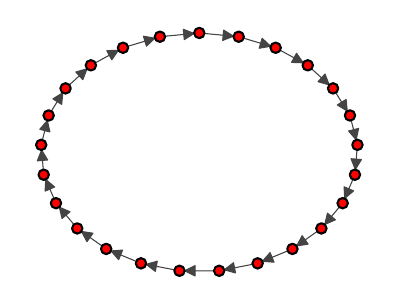}
\caption{The directed ring topology used in Sec.~\ref{sec:ex1}, with $K = 25$.}
\label{fig:ring}
\end{figure}

\begin{figure}
\centering
\includegraphics[width=0.3\linewidth]{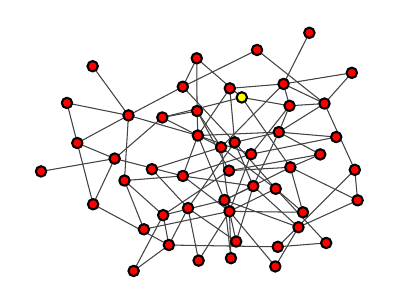}
\caption{One realization of the Erd\H{o}s-R\'enyi graph used in Sec.~\ref{sec:ex1}, with $c_K = 0.1 \ln(K)$ and $K = 50$.}
\label{fig:er-topology}
\end{figure}

\subsubsection{Ring topology} Let us consider the graph depicted in Fig.~\ref{fig:ring}. Observe that any agent $k$ has exactly one neighbor, one agent at distance $d_{jk}=1$, one agent at distance $2$, and so on, until the farthest agent, which is located at distance $K-2$ (recall that in our notation the distance is given by the number of hops \emph{minus $1$}). 
According to this description, the learning delay in \eqref{eq:arithmeticdelay} is given by
\begin{align}
T_k&=\frac 1 K \sum_{j=1}^K d_{jk}=
\frac 1 K\sum_{j=1}^{K-2} j=\frac{(K-1)(K-2)}{2\,K}
\nonumber\\
&=K\,
\left(1-\frac 1 K\right)\left(1-\frac 2 K\right),
\label{eq:lossex1ring}
\end{align}
which reveals that, over a ring topology, the learning delay for the decentralized system \emph{grows linearly with the network size}.

\subsubsection{Erd\H{o}s-R\'enyi graphs}
The ring topology is a very structured topology that does not favor inter-agent communication. We now switch to a graph model that lies somehow in the opposite extreme, namely, a \emph{random graph}. To avoid added complexity, we focus on the simplest generative mechanism, which is the Erd\H{o}s-R\'enyi model, where each edge is present with some probability. Since we are interested in connected networks, we focus on the regime where the graph is connected with probability that approaches $1$ as $K\rightarrow\infty$. It is known that this regime is characterized by a connection probability that scales with $K$ as~\cite{ErdosRenyi,RGbook}
\begin{equation}
\rho_K=\frac{\ln K + c_K}{K}, 
\end{equation}
where $c_K$ is any sequence that tends to $\infty$ as $K\rightarrow\infty$.

\begin{figure}
\centering
\includegraphics[width=0.5\linewidth]{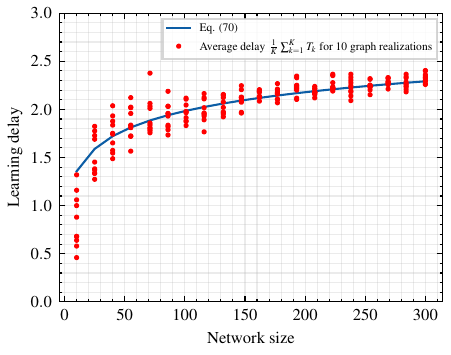}
\caption{Learning delay under Erd\H{o}s-R\'enyi network topologies as $K$ varies. $c_K = 0.1 \ln(K)$, $\Delta = 0.05$.}
\label{fig:er-performance-loss}
\end{figure}

\begin{figure}
\centering
\includegraphics[width=0.5\linewidth]{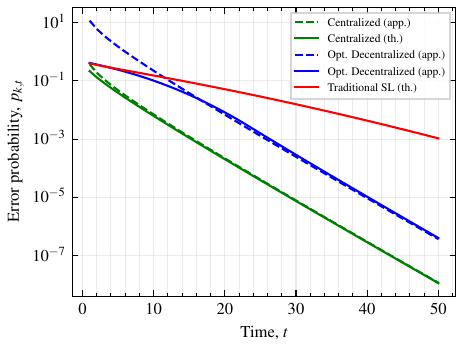}
\caption{Error probabilities of the considered decentralized decision-making strategies for the problem in Sec.~\ref{sec:ex1}, with $\Delta = 0.05$. The agents are connected according to the ring topology shown in Fig.~\ref{fig:ring}.}
\label{fig:ringproberr}
\end{figure}

\begin{figure}
\centering
\includegraphics[width=0.5\linewidth]{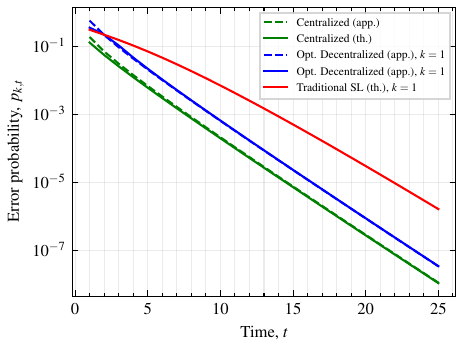}
\caption{Error probabilities of the considered decentralized decision-making strategies for the problem in Sec.~\ref{sec:ex1},  with $\Delta = 0.05$. The agents are connected according to the Erd\H{o}s-R\'enyi topology shown in Fig.~\ref{fig:er-topology}.}
\label{fig:ERproberr}
\end{figure}
In order to capture the scaling law of the loss with the number of agents, we resort to the following approximation for the typical distance between two network nodes~\cite{APLphysrevE}:
\begin{equation}
d_{jk} \approx \frac{\ln K - \gamma}{\ln (K \rho_K)}-\frac 1 2,
\label{eq:apptypdist}
\end{equation}
where $\gamma$ is Euler's constant. 
Substituting \eqref{eq:apptypdist} into \eqref{eq:arithmeticdelay}, we obtain
\begin{equation}
T_k\approx \frac{\ln K - \gamma}{\ln (K \rho_K)}-\frac 1 2.
\label{eq:ERapproxperf}
\end{equation}
Considering some slowly increasing sequence, e.g., $c_K\propto \ln K$, we see that the main scaling law for the learning delay $T_k$, under an Erd\H{o}s-R\'enyi graph, is approximately logarithmic with the network size, which is exponentially better than the law characterizing the ring topology. 
To check the goodness of the approximation in \eqref{eq:apptypdist}, in Fig.~\ref{fig:er-performance-loss} we plot the theoretical learning delay \eqref{eq:ERapproxperf} along with the learning delay averaged over the agents, for $10$ realizations of random graphs. We see that the random points follow well the logarithmic trend captured by \eqref{eq:ERapproxperf}.

In summary, by examining the ring and the random Erd\H{o}s-R\'enyi graphs, we find that: $i)$ over a ring topology, the learning delay caused by the decentralization grows linearly with the network size, while $ii)$ over an Erd\H{o}s-R\'enyi random graph, the significant reduction in the inter-agent distances (which scale logarithmically with the network size) yields an exponential reduction in the learning delay.

To gain further insight, in Fig.~\ref{fig:ringproberr} (ring topology) and Fig.~\ref{fig:ERproberr} (Erd\H{o}s-R\'enyi graph), we display the error probabilities for the optimal centralized decision system, the optimal decentralized decision system, and the traditional social learning strategy described in Sec.~\ref{sec:tradSL}. For the ring topology, we see from Fig.~\ref{fig:ringproberr} that the best decentralized system exhibits a loss of about two orders of magnitude with respect to the centralized MAP. Consistently with our theoretical analysis, this gap is \emph{not} reduced as $t$ increases; in fact, the curves pertaining to the centralized and decentralized systems stay parallel, showing the constant learning delay predicted by Corollary~\ref{th:corolla2}. Moreover, as $t$ increases, the predictions for the error probabilities in Theorem~\ref{th:mainfinalth} (dashed line) match the exact probabilities (solid line). 
The red curve depicts the performance achieved by traditional social learning. As revealed by \eqref{eq:tradSLvsdecloss}, traditional social learning exhibits a finite gap with respect to the best decentralized strategy, and we see that this gap is not negligible. All the error probabilities are evaluated exactly in this Gaussian problem~\cite{ourEUSIPCO2025}. 

The considerations made for the ring topology are essentially the same for the random graph in Fig.~\ref{fig:ERproberr}. Comparing this figure against Fig.~\ref{fig:ringproberr}, the main difference is that the loss experienced by the optimal decentralized scheme is reduced, which is consistent with the fact that the learning delay for the Erd\H{o}s-R\'enyi graph shown in \eqref{eq:ERapproxperf} is slower than the delay for the ring topology in \eqref{eq:lossex1ring}.

\begin{figure}
\centering
\includegraphics[width=0.5\linewidth]{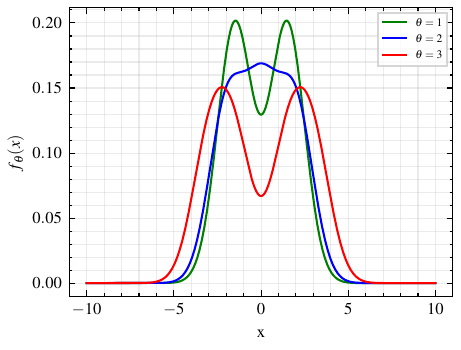}
\caption{Family of pdfs used in the decision problem from Sec.~\ref{sec:ex2}.}
\label{fig:mixt}
\end{figure}

\begin{table}[t]\small
\centering
\begin{tabular}{c|c|c|c|}
\cline{2-4}&
\multicolumn{3}{|c|}{
{\bf likelihood model} $\ell_j(\cdot|\theta)$}\\
\hline
\multicolumn{1}{|c|}{{\bf agent}}& 
$\theta=1$ & 
$\theta=2$ & 
$\theta=3$
\\
\hline
\multicolumn{1}{|c|}{$j=1,2,\ldots,45$}&
$f_1$ & $f_2$ & $f_2$
\\
\hline
\multicolumn{1}{|c|}{$j=46,47,\ldots,50$}&
$f_1$ & $f_2$ & $f_3$\\
\hline
\end{tabular}
\vspace{5pt}
\caption{Likelihood assignment for the example in Sec.~\ref{sec:ex2}.}
\label{tab:ident}
\end{table}

\subsection{Interplay Between Inference and Topology}
\label{sec:ex2}
In this section, we consider a decision problem where $K=50$ agents feature different levels of informativeness. Specifically, the decision problem consists of three possible hypotheses (namely, we set $\Theta=\{1,2,3\}$), which are assumed to be uniformly distributed a priori. 
The statistical distributions giving rise to the data are chosen from a family of Gaussian mixtures. Specifically, there are $3$ possible pdfs, which differ by the number of components and weights in the mixture. Specifically, denoting by $g(x)$ the pdf of the standard normal, we have the following three pdfs: 
\begin{align}
f_1(x)&=\frac 1 2 \, g(x+1.5) + \frac 1 2 \, g(x-1.5),\nonumber\\
f_2(x)&=\frac 1 3 \,g(x+2) + \frac 1 3 \,g(x) + \frac 1 3 \,g(x-2),\nonumber\\
f_3(x)&=\frac 1 4 \,g(x+3) + \frac 1 4 \,g(x+1.5) \nonumber\\
&+ \frac 1 4 \,g(x-1.5) + \frac 1 4 \,g(x-3),
\end{align}
which are displayed in Fig.~\ref{fig:mixt}.
The agents' likelihoods are chosen from these pdfs, and we assume that not all hypotheses are distinguishable for all agents. Specifically, the agents are clustered into two groups, corresponding to the likelihood assignment in Table~\ref{tab:ident}. 
For the first $45$ agents, we see from the table that, when $\theta=1$, the likelihood is given by $f_1$, while the likelihoods corresponding to $\theta=2$ and $\theta=3$ are equal to $f_2$. Accordingly, the agents in the first group are able to distinguish hypothesis $1$ from $2$ and $3$, but they confuse $2$ and $3$. 
The remaining $5$ agents are more informative, as they are able to distinguish all the hypotheses. 
The goal of this section is to show how formula \eqref{eq:overalloss} quantifies the interplay between the agents' informativeness and the network topology. 
To this end, we will examine the two topologies represented in Figs.~\ref{fig:ex2net1} and~\ref{fig:ex2net2}, respectively. 

\begin{figure}
\centering
\includegraphics[width=0.5\linewidth]{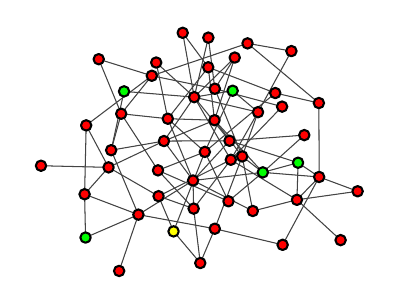}
\caption{The first topology used for the example in Sec.~\ref{sec:ex2}, with the most informative agents highlighted in green (5 out of 50 agents). Agent $1$, whose performance is displayed in Fig.~\ref{fig:perf1ex2}, is highlighted in yellow.}
\label{fig:ex2net1}
\end{figure}

\begin{figure}
\centering
\includegraphics[width=0.5\linewidth]{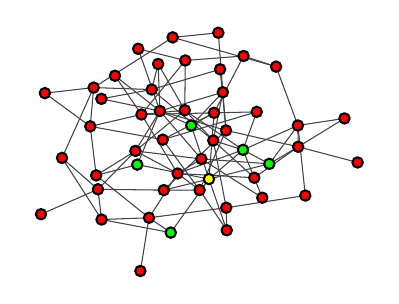}
\caption{The second topology used for the example in Sec.~\ref{sec:ex2}, with the most informative agents highlighted in green (5 out of 50 agents). Agent $1$, whose performance is displayed in Fig.~\ref{fig:perf1ex2}, is highlighted in yellow.}
\label{fig:ex2net2}
\end{figure}

In order to compute the losses $\mathscr{L}_k(\thetatrue,\theta)$ defined by \eqref{eq:lossth1} and appearing in \eqref{eq:overalloss}, we need to evaluate first the solution $s(\thetatrue,\theta)$ to the stationary equation \eqref{eq:stateqth}, and then the LMGFs
\begin{equation}
\Lambda_j(s(\thetatrue,\theta);\thetatrue,\theta).
\end{equation}
Then, after computing $\mathscr{L}_k(\thetatrue,\theta)$ for all pairs $(\thetatrue,\theta)$, we also evaluate the weights \eqref{eq:qqweights}, and finally obtain \eqref{eq:overalloss}.
Actually, in this example the LMGFs of the log likelihood ratios do not admit a closed form. Therefore, all the pertinent calculations have been performed numerically.

We are now ready to examine how the topologies in Figs.~\ref{fig:ex2net1} and~\ref{fig:ex2net2} affect performance. Specifically, for illustrative purposes we focus on the error probability of agent $1$ (represented by the circle highlighted in yellow).
\begin{figure}
\centering
\includegraphics[width=0.5\linewidth]{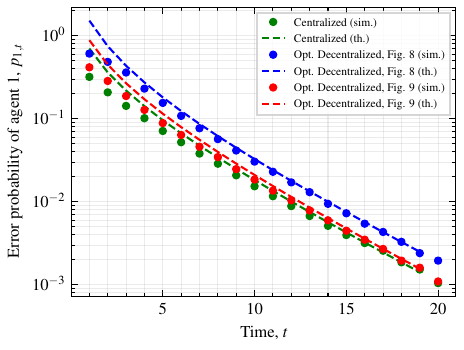}
\caption{Performance loss experienced by agent $1$ under the decision problem considered in Sec.~\ref{sec:ex2}, for the two topologies shown in Figs.~\ref{fig:ex2net1} and~\ref{fig:ex2net2}. The error probabilities are computed over $10000$ Monte Carlo trials.}
\label{fig:perf1ex2}
\end{figure}
Inspecting jointly Table~\ref{tab:ident} and the two topologies, we see that: $i)$ in Fig.~\ref{fig:ex2net1} the most informative agents (highlighted in green) are relatively distant from agent $1$, which is basically surrounded by the partially informative agents $j=1,2,\ldots,45$; while $ii)$ in Fig.~\ref{fig:ex2net2}, agent $1$ is surrounded by the most informative agents. Let us see how this difference affects the performance loss experienced by agent $1$, which is shown in Fig.~\ref{fig:perf1ex2}. 
We see that the loss is higher for the graph in Fig.~\ref{fig:ex2net1}. This is because, in that graph, the most informative agents are distant, which means that more information is lost due to decentralization. 
In comparison, for the topology in Fig.~\ref{fig:ex2net2} no loss is associated with the most informative agents, since they are immediate neighbors of agent $1$. 

Remarkably, the behavior observed empirically matches well the theoretical results. In fact, the analytical characterization provided by Theorems~\ref{th:theorem1maintext} and~\ref{th:mainfinalth} offer reliable formulas (see the dashed curves in Fig.~\ref{fig:perf1ex2}) that are able to quantify precisely how the agents' informativeness and connectivity interact and impact the final error probability performance.

\section{Conclusion and Future Work}
In this work, we established a fundamental limit on the error probability performance for decentralized decision systems. This limit probability admits a closed form that highlights the role of the network topology and the decision problem. The former appears in the analytical formulas through the shortest distances between pairs of agents. The latter is represented by quantities relative to the moment generating functions of the random observations collected by the agents. The analysis revealed that, contrary to a widespread belief, the performance of \emph{any} decentralized decision strategy \emph{cannot reach} the performance of the corresponding centralized system. Remarkably, even in the long run, i.e., over infinite streams of data, \emph{this gap persists}. 
We also showed that the gap in terms of error probability translates into a \emph{learning delay} that slows down the decentralized system irremediably.

We exploited the aforementioned theoretical formulas to gain insight into useful network structures. 
First, to magnify the impact of the network topology on the decentralization loss, we considered a scenario where all agents are equally informative. We established that in this case the performance loss scales exponentially with the sum of the network distances. Accordingly, we showed that the error probability for the best decentralized strategy can be larger than the optimal centralized error probability \emph{by orders of magnitude}.
We examined two standard network constructions, namely, the ring topology and the Erd\H{o}s-R\'enyi random graphs. 
In the former case, the learning delay of any agent \emph{grows linearly with the number of agents}. This is due to the fact that, over a ring network, the average distance between two agents scales linearly with the number of agents. Over the Erd\H{o}s-R\'enyi graphs, the learning delay is significantly reduced, namely, it \emph{grows logarithmically with the number of agents}. This is consistent with the fact that, according to the theory of random graphs, the average distance between two nodes over an Erd\H{o}s-R\'enyi topology scales logarithmically with the network size.

Second, we used the theoretical formulas to quantify the interplay between the topology and the agents' informativeness. While it might be intuitive that it is beneficial for an agent to be closer to the most informative agents, the evaluation of the error probability depends in a nontrivial manner on the interplay between the agents' connectivity (quantified by the distances over the graph) and the level of informativeness (quantified by the moment generating functions). We showed that our formulas are able to provide reliable predictions for the agents' performance under different conditions where the position of the most informative agents in the network is varied.

The derived performance limit should be used as a benchmark for the performance of any decentralized strategy. 
For example, we showed that traditional social learning schemes exhibits a finite gap with respect to the optimal decentralized strategy. In this connection, one useful extension would be the analytical evaluation of the performance achieved by existing decision-making schemes, to quantify how far they are from the best attainable decentralized performance.  

Another future research direction is to design decentralized decision schemes under different constraints (e.g., communication constraints) and exploit the rate-distortion tradeoff to quantify how the gap with respect to the best decentralized performance varies with the available communication resources. 

Finally, our analysis leads to a useful comparison between the two main inferential problems, namely, classification and estimation. In the former case, one deals with categorical variables (the hypotheses), for which it is necessary to employ a hit-or-miss metric like the error probability. In the latter problem, one can resort to error measures suited for real values, like the mean-square-error. For decentralized estimation problems, the optimal mean-square-error can be asymptotically achieved. In contrast, we have proved that the optimal error probability cannot be achieved by any decentralized strategy. 
This highlights an intrinsic difference between decision-making and estimation. 
At the same time, the results in the present work suggest that, under interval estimation (where the error probability is used for the performance evaluation of an estimator) the asymptotic optimality observed in the mean-square-error sense needs to be revisited.

\appendices

\section{Exact Asymptotics of Certain Random Sums}
\label{app:asysums}
In the following theorems, we will consider a doubly indexed sequence of random variables $\bm{\xi}_{j,\tau}$, for $j\in\mathcal{J}$ and $\tau\in\mathbb{N}$, which fulfill the following assumption.

\begin{assumption}[\textbf{Regularity conditions on} $\bm{\xi}_{j,\tau}$]
\label{assum:csijtau}
The random variables $\bm{\xi}_{j,\tau}$ are absolutely continuous with respect to the Lebesgue measure on $\mathbb{R}$, and statistically independent (across both indices $j$ and $\tau$). 
For all $\tau$, they share the same pdf $f_j$, with
\begin{equation}
\mathbb{E}\,\bm{\xi}_{j,\tau}>0,\qquad \inf \mathrm{Supp}(f_j)<0,
\label{eq:posnegcond}
\end{equation}
where $\mathrm{Supp}(f_j)$ is the support of $f_j$~\cite[Def. E.1]{MattaBordignonSayedBook}. For each $f_j$, the characteristic function 
\begin{equation}
\varphi_j(s)\triangleq \mathbb{E} \exp\big(\iota\, s\, \bm{\xi}_{j,\tau}\big)
\end{equation} (with $s\in\mathbb{R}$ and $\iota=\sqrt{-1}$) fulfills the following bound:
\begin{equation}
|\varphi_j(s)| \leq \frac{c}{|s|^\delta},\qquad |s|\geq r,
\label{eq:boundonchf}
\end{equation}
for some positive constants $\delta, c,$ and $r$.
Moreover, the log moment generating function 
\begin{equation}
\Lambda_j(s)\triangleq \ln \mathbb{E} \exp\big(s\, \bm{\xi}_{j,\tau}\big)
\end{equation} 
is finite for all $s\in\mathbb{R}$.
\hfill$\square$
\end{assumption}
Equation \eqref{eq:boundonchf} represents a mild condition that is verified, for instance, when the pdf $f_j$ has an integrable first derivative~\cite[Lemma 4, Chap. XV]{Feller2}. 
Regarding the finiteness of the LMGF, some examples satisfying this condition are: all random variables with finite support, the Gaussian distribution, or any Gaussian mixture.

Before stating the main result of this appendix, it is useful to introduce the aggregate variable
\begin{equation}
\bm{\xi}_\tau\triangleq\sum_{j\in\mathcal{J}}\bm{\xi}_{j,\tau} ,\quad \tau\in\mathbb{N},
\label{eq:barcsitau}
\end{equation}
along with its LMGF
\begin{equation}
\Lambda(s)\triangleq \ln \mathbb{E}\,\exp\big(s \, \bm{\xi}_\tau\big)= \sum_{j\in\mathcal{J}} \Lambda_j(s),
\end{equation}
where the equality follows from the additivity of the LMGF for independent variables. Note that, since each $\bm{\xi}_{j,\tau}$ is nondeterministic (being absolutely continuous with respect to the Lebesgue measure on $\mathbb{R}$), then each LMGF $\Lambda_j(s)$ is strictly convex, which implies that so is $\Lambda(s)$~\cite{DemboZeitouni,DenHollander,MattaBordignonSayedBook}. 
Now, it is known that the limit of the first derivative\footnote{We recall that, since by assumption the LMGF $\Lambda_j(s)$ is finite for all $s\in\mathbb{R}$, it is infinitely differentiable on $\mathbb{R}$~\cite{MattaBordignonSayedBook}.} $\Lambda^{\prime}(s)$ as $s\rightarrow -\infty$ (resp., $s\rightarrow +\infty$) is equal to the infimum (resp., the supremum) of the support of the pdf of $\bm{\xi}_\tau$~\cite[proof of Lemma E.1]{MattaBordignonSayedBook}. On the other hand, from \eqref{eq:posnegcond} we know that the infimum of the support of $f_j$ is negative, while the supremum is positive (because the mean $\mathbb{E}\,\bm{\xi}_{j,\tau}$ is positive). 
This implies that the stationary equation
\begin{equation}
\Lambda^{\prime}(s_0)=\sum_{j\in\mathcal{J}} \Lambda_j^{\prime}(s_0)=0 
\label{eq:stateq}
\end{equation}
has always a solution, which is the unique minimizer of $\Lambda(s)$ in view of the strict convexity.

Moreover, since the first derivative of the LMGF evaluated in $0$ is equal to the mean, we have $\Lambda^{\prime}(0)>0$. This implies that the minimizer of $\Lambda(s)$ is located on the negative axis. Moreover, since $\Lambda(0)=0$ by definition, the minimum of $\Lambda(s)$ is negative. In summary, we have
\begin{equation}
s_0<0,\qquad \Lambda(s_0)=\sum_{j\in\mathcal{J}} \Lambda_j(s_0)<0.
\end{equation}

\begin{theorem}[\textbf{Exact asymptotics of some useful random sums}]
\label{th:theorapp}
Let the random variables $\bm{\xi}_{j,\tau}$, for $j\in\mathcal{J}$ and $\tau\in\mathbb{N}$, satisfy Assumption~\ref{assum:csijtau}. For $j\in\mathcal{J}$, let $d_{j}$ be nonnegative integer numbers. 
For $t>\max\limits_{j\in\mathcal{J}}{d_j}$, define the random sequence \begin{equation}
\bm{z}_t=
\sum_{j\in\mathcal{J}} \sum_{\tau=1}^{t-d_j} \bm{\xi}_{j,\tau} 
+ \gamma,
\label{eq:ztdefTh1}
\end{equation}
where $\gamma\in\mathbb{R}$. 
Let $s_0<0$ be the solution to the stationary equation \eqref{eq:stateq} and define
\begin{equation}
\sigma^2_t\triangleq s_0^2\,t\,\sum_{j\in\mathcal{J}} \Lambda_j^{\prime\prime}(s_0)
\label{eq:sigmatdef}
\end{equation}
and
\begin{equation}
p_t\triangleq 
\frac{1}{\sqrt{2\pi\,\sigma^2_t}}
\exp\Bigg(t\,\sum_{j\in\mathcal{J}}\Lambda_j(s_0)
- \sum_{j\in\mathcal{J}} d_j \Lambda_j(s_0)
-|s_0|\,\gamma  
\Bigg).
\label{eq:ptdef}
\end{equation}
Then,
\begin{equation}
\lim_{t\rightarrow\infty}
\frac
{\mathbb{P}[\bm{z}_t\leq 0]}
{p_t}=1.
\label{eq:claimth1}
\end{equation}
\end{theorem}

\begin{IEEEproof}
Let us consider the change of measure identified by the so-called \emph{tilted pdf}~\cite[App. E]{MattaBordignonSayedBook}
\begin{equation}
\tilde{f}_j(\xi)=f_j(\xi)\,
\frac{e^{s_0 \, \xi}}{e^{\Lambda_j(s_0)}},\qquad 
j\in\mathcal{J},\; \xi\in\mathbb{R}.
\label{eq:tiltedef}
\end{equation}
From the definition of the LMGF, it is easily verified that $\tilde{f}_j$ is a pdf. Moreover, it is known that~\cite[App. E]{MattaBordignonSayedBook} the first three moments (which will be relevant later in our proof) of $\bm{\xi}_{j,\tau}$ under the tilted distribution are given by
\begin{align}
\E_{\tilde{f}_j}\, \bm{\xi}_{j,\tau}&=\Lambda_j^{\prime}(s_0),
\label{eq:tildedmu1}
\\
\E_{\tilde{f}_j} 
\left[
\left(\bm{\xi}_{j,\tau} - \E_{\tilde{f}_j}\, \bm{\xi}_{j,\tau}\right)^2
\right]&=\Lambda_j^{\prime\prime}(s_0),
\label{eq:tildedmu2}
\\
\E_{\tilde{f}_j} 
\left[
\left(\bm{\xi}_{j,\tau} - \E_{\tilde{f}_j}\, \bm{\xi}_{j,\tau}\right)^3
\right]&=\Lambda_j^{\prime\prime\prime}(s_0)
\label{eq:tildedmu3},
\end{align}
where the subscript $\tilde{f}_j$ denotes that the expectation is computed under the tilted distribution. 

We have the following chain of identities:
\begin{align}
\mathbb{P}[\bm{z}_t\leq 0]&=
\mathbb{E}\big[\mathbb{I}_{\bm{z}_t\leq 0}\big]\nonumber\\
&\overset{\mathrm{(a)}}{=}
\int \mathbb{I}_{z_t\leq 0}\,
\left(\prod_{j\in\mathcal{J}}\prod_{\tau=1}^{t-d_j}
f_j(\xi_{j,\tau})\,d\xi_{j,\tau}
\right)
\nonumber\\
&\overset{\mathrm{(b)}}{=}\int \mathbb{I}_{z_t\leq 0}\left(\prod_{j\in\mathcal{J}}\prod_{\tau=1}^{t-d_j}
\frac{e^{\Lambda_j(s_0)}}
{e^{s_0 \, \xi_{j,\tau}}}
\tilde{f}_j(\xi_{j,\tau})\,d\xi_{j,\tau}\right)
\nonumber\\
&\overset{\mathrm{(c)}}{=}
\mathbb{E}_{\tilde{f}}\left[
\frac{
\exp\left(
\sum_{j\in\mathcal{J}}\sum_{\tau=1}^{t-d_j}
\Lambda_j(s_0)
\right)
}
{\exp\left(
s_0 \, 
\sum_{j\in\mathcal{J}}\sum_{\tau=1}^{t-d_j}\bm{\xi}_{j,\tau}
\right)}\;\;
\mathbb{I}_{\bm{z}_t\leq 0}\right]
\nonumber\\
&\overset{\mathrm{(d)}}{=}\mathbb{E}_{\tilde{f}}\left[
\frac{
\exp\left(
\sum_{j\in\mathcal{J}}(t-d_j)\,
\Lambda_j(s_0)
\right)
}
{\exp\big(
s_0 \, (\bm{z}_t - \gamma) 
\big)}\;\;
\mathbb{I}_{\bm{z}_t\leq 0}\right],
\label{eq:fundamentaltrick}
\end{align}
where (a) holds because the random variables $\bm{\xi}_{j,\tau}$ are independent across $j$ and $\tau$; (b) follows from 
\eqref{eq:tiltedef}; in (c) we introduce the subscript $\tilde{f}$ to denote the expectation computed under the tilted measure associated with the product of the tilted pdfs $\tilde{f}_j(\xi_{j,\tau})$; and (d) follows from computing the sum over $\tau$ in the numerator (note that the LMGFs $\Lambda_j$ do not depend on $\tau$) and from applying \eqref{eq:ztdefTh1} in the denominator. 
Pulling the deterministic quantities out of the expectation operator, we can rewrite \eqref{eq:fundamentaltrick} as
\begin{align}
\mathbb{P}[\bm{z}_t\leq 0]&=\mathbb{E}_{\tilde{f}}\Big[
e^{-s_0\,\bm{z}_t}\,
\mathbb{I}_{\bm{z}_t\leq 0}
\Big]
\nonumber\\
&\times
\exp\left(t\sum\limits_{j\in\mathcal{J}} \Lambda_j(s_0)
-
\sum\limits_{j\in\mathcal{J}} d_j \Lambda_j(s_0) - 
|s_0|\,\gamma\right),
\label{eq:Pzt0}
\end{align}
where we used the fact that $s_0<0$.

Consider now an integer number $n>d_j$. 
It is convenient to split the random variable $\bm{z}_t$ in \eqref{eq:ztdefTh1} as follows:
\begin{align}
\bm{z}_t&=
\sum_{j\in\mathcal{J}} \sum_{\tau=1}^{t-n} \bm{\xi}_{j,\tau}
+
\sum_{j\in\mathcal{J}} \sum_{\tau=t-n+1}^{t-d_j} \bm{\xi}_{j,\tau}
+ \gamma\nonumber\\
&=
\sum_{\tau=1}^{t-n} \bm{\xi}_{\tau} + \sum_{j\in\mathcal{J}} \sum_{\tau=t-n+1}^{t-d_j} \bm{\xi}_{j,\tau}
+ \gamma,
\label{eq:zsplitnew}
\end{align}
where in the last equality we applied the definition of $\bm{\xi}_\tau$ from \eqref{eq:barcsitau}.
By introducing a random variable $\bm{\xi}_0$,  statistically independent of the variables $\bm{\xi}_{j,\tau}$ and such that
\begin{equation}
\bm{\xi}_0\stackrel{\mathrm{d}}{=}\sum_{j\in\mathcal{J}} \sum_{\tau=t-n+1}^{t-d_j} \bm{\xi}_{j,\tau}
+ \gamma,
\label{eq:barcsi0}
\end{equation}
where $\stackrel{\mathrm{d}}{=}$ denotes equality in distribution, we get
\begin{equation}
\bm{z}_t\stackrel{\mathrm{d}}{=}
\bm{\xi}_0 + \sum_{\tau=1}^{t-n}\bm{\xi}_\tau=
\sum_{\tau=0}^{t-n} \bm{\xi}_\tau.
\label{eq:distequality}
\end{equation}
Observe that, in view of \eqref{eq:stateq} and \eqref{eq:tildedmu1}, the means, under the tilted distribution, of the random variables defined by \eqref{eq:barcsitau} are zero.
On the other hand, from \eqref{eq:barcsi0} we see that the mean of $\bm{\xi}_0$ under the tilted distribution is given by
\begin{equation}
\mathbb{E}_{\tilde{f}} \,\bm{\xi}_0=
n \underbrace{\sum_{j\in\mathcal{J}} \Lambda_j^{\prime}(s_0)}_{=0 \textnormal{ from \eqref{eq:stateq}}}
-
\sum_{j\in\mathcal{J}} d_j \Lambda_j^{\prime}(s_0)+\gamma. 
\label{eq:meanvartransform}
\end{equation}
Therefore, we can write
\begin{equation}
m_z\triangleq \mathbb{E}_{\tilde{f}}\,\bm{z}_t=-\sum_{j\in\mathcal{J}} d_j \Lambda_j^{\prime}(s_0)+\gamma.
\label{eq:mzdef}
\end{equation}
Likewise, by exploiting \eqref{eq:tildedmu1} and the independence of the random variables $\bm{\xi}_{j,\tau}$, we obtain (the notation for the variance should be self-explaining),
\begin{equation}
\sigma^2_{z,t}\triangleq \mathbb{VAR}_{\tilde{f}}\, \bm{z}_t =
t \sum_{j\in\mathcal{J}} \Lambda_j^{\prime\prime}(s_0)
-
\sum_{j\in\mathcal{J}} d_j \, \Lambda_j^{\prime\prime}(s_0).
\label{eq:sigmaztdef}
\end{equation}
Let us introduce the standardized version of $\bm{z}_t$,
\begin{equation}
\bm{\zeta}_t\triangleq \frac{\bm{z}_t - m_z}{\sigma_{z,t}},
\label{eq:standz}
\end{equation}
and denote by $\mathsf{F}_t$ its cumulative distribution function (cdf) under the tilted measure. From~\cite[Thm. 7, Chap. VI]{Petrov} (whose sufficient conditions can be readily verified to hold in view of the properties in Assumption~\ref{assum:csijtau}), it turns out that: $i)$ $\bm{\zeta}_t$ converges in distribution (under the tilted measure) to a standard normal; $ii)$ its cdf $\mathsf{F}_t$ can be approximated by the following first-order Edgeworth expansion: 
\begin{equation}
\widehat{\mathsf{F}}_t(\zeta)\triangleq G(\zeta) + \frac{m_{z,t}^{(3)}}{6\,\sigma_{z,t}^{3}}\,h(\zeta),
\label{eq:Edgeworthexpansionfirsthatdef}
\end{equation}
where
\begin{equation}
m^{(3)}_{z,t}\triangleq
\sum_{\tau=0}^{t-n} 
\E_{\tilde{f}} 
\left[
\left(\bm{\xi}_{\tau} - \E_{\tilde{f}}\, \bm{\xi}_{\tau}\right)^3
\right], \quad h(\zeta)\triangleq (1-\zeta^2)g(\zeta),
\label{eq:m3def}
\end{equation}
with $G$ and $g$ being the cdf and the pdf of the standard normal, respectively. 
The approximation holds in the following sense~\cite[Thm. 7, Chap. VI]{Petrov}:
\begin{equation}
\lim_{t\rightarrow\infty}
\sqrt{t} \,
\sup_{\zeta\in\mathbb{R}}
\left|
\mathsf{F}_t(\zeta) - \widehat{\mathsf{F}}_t(\zeta) 
\right|=0,
\label{eq:Edgeworth}
\end{equation}
that is, the error 
\begin{equation}
\mathsf{E}_t(\zeta)\triangleq 
\mathsf{F}_t(\zeta) - \widehat{\mathsf{F}}_t(\zeta) 
\label{eq:errFhatF}
\end{equation}
between $\mathsf{F}_t$ and the Edgeworth approximation $\widehat{\mathsf{F}}_t$ converges to zero faster than $1/\sqrt{t}$, and uniformly in $\zeta$.

Exploiting the definition of the standardized variable $\bm{\zeta}_t$ from  \eqref{eq:standz}, we have the identities (we recall that $s_0<0$)
\begin{align}
&\mathbb{E}_{\tilde{f}}
\left[
e^{-s_0\bm{z}_t}\,
\mathbb{I}_{\bm{z}_t\leq 0}
\right]=
e^{-s_0 m_z}
\,
\mathbb{E}_{\tilde{f}}
\left[
e^{-s_0 \sigma_{z,t} \bm{\zeta}_t}\,
\mathbb{I}_{\bm{\zeta}_t\leq - m_z/\sigma_{z,t}}
\right]\nonumber\\
&=
|s_0| \,\sigma_{z,t}\,
e^{|s_0| m_z}\!\!\!
\int_{-\infty}^{-m_z/\sigma_{z,t}}
\!\!\!\!\!\!\!\!\!\!\!\!\!\!\!\!
e^{|s_0| \sigma_{z,t} \zeta}
\left[
\mathsf{F}_t\left(-\frac{m_z}{\sigma_{z,t}}\right)
-
\mathsf{F}_t(\zeta)
\right]
d\zeta,
\label{eq:prefinalform}
\end{align}
where the last equality follows from the definition of the expected value and the rule of integration by parts. 
By using the change of variable $|s_0|(\sigma_{z,t}\,\zeta+m_z)\mapsto \zeta $, we get
\begin{align}
&\mathbb{E}_{\tilde{f}}
\left[
e^{-s_0\bm{z}_t}\,
\mathbb{I}_{\bm{z}_t\leq 0}
\right]
\nonumber\\
&=
\underbrace{
\int_{-\infty}^{0}
e^{\zeta}
\left[
\mathsf{F}_t\left(-\frac{m_z}{\sigma_{z,t}}\right)
-
\mathsf{F}_t\left(\frac{\zeta-|s_0|m_z}{|s_0|\,\sigma_{z,t}}\right)\right]
d\zeta}_{\triangleq I_t}.
\label{eq:withItdef}
\end{align}
Combining \eqref{eq:ptdef}, \eqref{eq:Pzt0}, and \eqref{eq:withItdef}, we can write
\begin{equation}
\frac{\mathbb{P}[\bm{z}_t\leq 0]}{p_t}=
\sqrt{2\pi}\,\sigma_t\,I_t.
\label{eq:finalim}
\end{equation}
Therefore, the claim in \eqref{eq:claimth1} is equivalent to
\begin{equation}
\lim_{t\rightarrow\infty}\sqrt{2\pi} \sigma_t \, I_t = 1.
\label{eq:equivclaimlimItnewlast}
\end{equation}
Let us then elaborate to prove \eqref{eq:equivclaimlimItnewlast}. Using \eqref{eq:errFhatF}, the quantity $I_t$ defined in \eqref{eq:withItdef} can be represented as
\begin{align}
I_t
&=\underbrace{
\int_{-\infty}^{0}
e^{\zeta}
\left[
\widehat{\mathsf{F}}_t\left(-\frac{m_z}{\sigma_{z,t}}\right)
-
\widehat{\mathsf{F}}_t\left(\frac{\zeta-|s_0|m_z}{|s_0|\,\sigma_{z,t}}\right)\right]
d\zeta}_{\triangleq I_{1,t}}
\nonumber\\
&+
\underbrace{
\int_{-\infty}^{0}
e^{\zeta}
\left[
\mathsf{E}_t\left(-\frac{m_z}{\sigma_{z,t}}\right)
-
\mathsf{E}_t\left(\frac{\zeta-|s_0|m_z}{|s_0|\,\sigma_{z,t}}\right)\right]
d\zeta}_{\triangleq I_{2,t}}.
\label{eq:firstappearanceofI1andI2}
\end{align}
Regarding the term $I_{2,t}$, from the triangle inequality we have the upper bound
\begin{align}
\left|I_{2,t}\right|
\leq 2\,\sup_{\zeta\in\mathbb{R}}|\mathsf{E}_t(\zeta)|,
\label{eq:newboundonI2t}
\end{align}
where we also exploited the identity $\int_{-\infty}^{0}e^{\zeta}\,d\zeta=1$. 
Applying the definition of $\sigma_t$ from \eqref{eq:sigmatdef}, we can write
\begin{equation}
\sqrt{2\pi} \sigma_t \, I_{2,t}=
\sqrt{t}\,I_{2,t}
\Bigg(
2\pi s_0^2\,\sum_{j\in\mathcal{J}} \Lambda_j^{\prime\prime}(s_0)
\Bigg)^{1/2}.
\label{eq:explicitI2tlimnew}
\end{equation}
From \eqref{eq:Edgeworth}, \eqref{eq:errFhatF}, \eqref{eq:newboundonI2t}, and \eqref{eq:explicitI2tlimnew}, we conclude that
\begin{equation}
\lim_{t\rightarrow\infty} \sqrt{2\pi} \sigma_t \, I_{2,t}=0.
\end{equation}
Thus, to prove \eqref{eq:equivclaimlimItnewlast}, it remains to show that 
\begin{equation}
\lim_{t\rightarrow\infty} \sqrt{2\pi} \sigma_t \, I_{1,t}=1.
\label{eq:newnewclaimlimI1tpartialterm}
\end{equation}
To this end, we exploit the definitions of $\widehat{\mathsf{F}}_t$ from \eqref{eq:Edgeworthexpansionfirsthatdef} and $I_{1,t}$ from \eqref{eq:firstappearanceofI1andI2} to write
\begin{align}
&\sqrt{2\pi}\,\sigma_t\,I_{1,t}
\nonumber\\
&=\sqrt{2\pi}\,\sigma_t\,\int_{-\infty}^0
\left[G\left(-\frac{m_z}{\sigma_{z,t}}\right)
-
G\left(\frac{\zeta-|s_0|m_z}{|s_0|\,\sigma_{z,t}}\right)
\right]
e^{\zeta}
d\zeta\nonumber\\
&+
\sqrt{\frac{\pi}{18}}
\frac{m_{z,t}^{(3)} \, \sigma_t}
{\sigma_{z,t}^{3}}\nonumber\\
&\times
\int_{-\infty}^0
\left[h\left(-\frac{m_z}{\sigma_{z,t}}\right)
-
h\left(\frac{\zeta-|s_0|m_z}{|s_0|\,\sigma_{z,t}}\right)
\right]
e^{\zeta}
d\zeta.
\label{eq:closetoend}
\end{align}
First, we focus on the second term. 
The ratio 
\begin{equation}
\frac{m_{z,t}^{(3)} \, \sigma_t}
{\sigma_{z,t}^{3}}
\end{equation}
converges to a finite limit as $t\rightarrow\infty$ because $m^{(3)}_{z,t}$ scales as $t$ in view of \eqref{eq:m3def}, while $\sigma_t$ and $\sigma_{z,t}$ scale as $\sqrt{t}$ in view of \eqref{eq:sigmatdef} and \eqref{eq:sigmaztdef}, respectively. Moreover, since
\begin{equation}
\lim_{t\rightarrow\infty }
\left[h\left(-\frac{m_z}{\sigma_{z,t}}\right)
-
h\left(\frac{\zeta-|s_0|m_z}{|s_0|\,\sigma_{z,t}}\right)
\right]
e^{\zeta}=0,
\end{equation}
and since the function $h$ is bounded, by applying the dominated convergence theorem we conclude that the second term in \eqref{eq:closetoend} vanishes. 
Let us move on to examining the first term in \eqref{eq:closetoend}. By applying the rule of integration by parts, we can rewrite that term as
\begin{equation}
\sqrt{2\pi} \frac{\sigma_t}{|s_0|\sigma_{z,t}} \int_{-\infty}^0
g\left(
\frac
{\zeta-|s_0|m_z}
{|s_0|\sigma_{z,t}}
\right)e^\zeta d\zeta.
\end{equation}
From \eqref{eq:sigmatdef} and \eqref{eq:sigmaztdef} we conclude that the ratio $\sigma_t/(|s_0|\sigma_{z,t})$ converges to $1$. By applying the dominated convergence theorem (observe that $g$ is bounded), we finally obtain
\begin{align}
&\lim_{t\rightarrow\infty}
\sqrt{2\pi} \frac{\sigma_t}{|s_0|\sigma_{z,t}} \int_{-\infty}^0
g\left(
\frac
{\zeta-|s_0|m_z}
{|s_0|\sigma_{z,t}}
\right)e^\zeta d\zeta
\nonumber\\
&=
\sqrt{2\pi} g(0) \int_{-\infty}^0
e^\zeta d\zeta=1,
\end{align}
which completes the proof.
\end{IEEEproof}

\section{Pairs of Log Likelihood Ratios}
\label{app:propLLR}

\begin{lemma}[\textbf{Non-degeneracy of the global LLRs}]
\label{lem:LLRnondeg}
Let Assumptions~\ref{assum:luckylike} and~\ref{assum:ident} be satisfied.
Consider the true hypothesis $\thetatrue$, and two other hypotheses $\theta_1$ and $\theta_2$. 
Then, the global log likelihood ratios $\bm{\lambda}_t(\thetatrue,\theta_1)$ and $\bm{\lambda}_t(\thetatrue,\theta_2)$ are not proportional to each other.
\end{lemma}
\begin{IEEEproof}
It is convenient to introduce the vector $x=[x_{1},x_{2},\ldots,x_{K}]$, which stacks the vectors $x_k\in\mathbb{R}^{n_k}$ for $k\in\mathcal{K}$. 
By introducing the \emph{global} likelihood
\begin{equation}
\ell(x|\theta)\triangleq \prod_{k=1}^K \ell_k(x_k|\theta),\quad \theta\in\Theta,
\label{eq:globlikeli}
\end{equation}
the claim of the lemma is equivalent to showing that $\ln\frac{\ell(x|\thetatrue)}{\ell(x|\theta_1)}$ and $\ln\frac{\ell(x|\thetatrue)}{\ell(x|\theta_2)}$ are not proportional to each other.
By contrapositive statement, assume that
\begin{equation}
\ln\frac{\ell(x|\thetatrue)}{\ell(x|\theta_1)}=
a\,\ln\frac{\ell(x|\thetatrue)}{\ell(x|\theta_2)}
\label{eq:proportione}
\end{equation}
for some $a\in\mathbb{R}$ and all $x\in\mathbb{R}^{\sum_{k=1}^K n_k}$, but for a zero probability set under $\ell(x|\thetatrue)$. Now, the constant $a$ cannot be equal to $0$ or $1$ because otherwise we would violate the identifiability assumption (in particular, when $a=1$ we would violate the distinguishability between $\theta_1$ and $\theta_2$). 
On the other hand, $a$ cannot be negative, otherwise (since identifiability corresponds to positivity of the KL divergence), we would have
\begin{equation}
0<\sum_{k=1}^K D_k(\thetatrue||\theta_1)=a\,\sum_{k=1}^K D_k(\thetatrue||\theta_2)<0.
\end{equation}
Accordingly the residual possibility is that $a>0$ with $a\neq 1$. 
If Eq. \eqref{eq:proportione} holds, we have
\begin{equation}
\ell(x|\theta_1)=\ell^a(x|\theta_2)\,\ell^{1-a}(x|\thetatrue).
\label{eq:proporconjec}
\end{equation}
Assume that $0<a<1$, and set
\begin{equation}
u(x)=\ell^a(x|\theta_2),\quad 
v(x)=\ell^{1-a}(x|\thetatrue),
\label{eq:fgdef}
\end{equation}
and
\begin{equation}
p=\frac{1}{a}, \qquad q=\frac{1}{1-a}. 
\label{eq:pqdef}
\end{equation}
In view of Hölder's inequality
\begin{equation}
\int u(x)\,v(x) \,dx \leq 
\left(\int u^p(x) \,dx\right)^\frac{1}{p}
\!\!
\left(\int v^q(x) \,dx\right)^\frac{1}{q}
\!\!\!,
\end{equation}
which, by substituting \eqref{eq:fgdef} and \eqref{eq:pqdef}, yields
\begin{equation}
\int \ell^a(x|\theta_2)\,\ell^{1-a}(x|\thetatrue)\,dx\leq 1.
\end{equation}
However, if \eqref{eq:proporconjec} holds, the above inequality should in fact be an equality, which, by Hölder's result, holds iff 
\begin{equation}
u^p(x)\propto v^q(x),
\end{equation}
that is, iff
\begin{equation}
\ell(x|\theta_2)\propto \ell(x|\thetatrue),
\label{eq:Holderllrprop}
\end{equation}
which is impossible in view of identifiability.\footnote{
Since $\int\ell(x|\theta_2)dx=\int\ell(x|\thetatrue)dx$, Eq. \eqref{eq:Holderllrprop} would correspond to $\ell(\cdot|\theta_2)=\ell(\cdot|\thetatrue)$.}
The case $a>1$ can be handled similarly, by considering $1/a$ in place of $a$ and exchanging the role of $\theta_1$ and $\theta_2$.
\end{IEEEproof}

\begin{lemma}[\textbf{Strict convexity of the bivariate LMGF}]
\label{lem:LMGFsc}
Let Assumptions~\ref{assum:luckylike}--\ref{assum:finiteMGFs} be satisfied, and consider the bivariate LMGF
\begin{align}
&\Lambda(s_1,s_2;\thetatrue,\theta_1,\theta_2)\nonumber\\
&\triangleq
\ln\mathbb{E}_{\thetatrue}\, \exp
\Big(
s_1\,\bm{\lambda}_t(\thetatrue,\theta_1)
+
s_2\,\bm{\lambda}_t(\thetatrue,\theta_2)
\Big).
\label{eq:bivLMGFlikedef}
\end{align}
Then, $\Lambda(s_1,s_2;\thetatrue,\theta_1,\theta_2)$ is strictly convex.
\end{lemma}
\begin{IEEEproof}
For ease of notation, in the following we skip the explicit dependence of the pertinent quantities on $\thetatrue$, $\theta_1$, and $\theta_2$.
Consider the bivariate moment generating function (MGF)
\begin{equation}
M(s_1,s_2)
\triangleq
\mathbb{E}_{\thetatrue} \exp
\Big(
s_1\,\bm{v}_1
+
s_2\,\bm{v}_2
\Big),
\end{equation}
where we set
\begin{equation}
\bm{v}_1\triangleq \bm{\lambda}_t(\thetatrue,\theta_1),\qquad
\bm{v}_2\triangleq \bm{\lambda}_t(\thetatrue,\theta_2).
\label{eq:v1v2def}
\end{equation}
In view of Assumption~\ref{assum:finiteMGFs}, the marginal MGFs for $\bm{v}_1$ and $\bm{v}_2$, say $M_1(s_1)$ and $M_2(s_2)$, exist and are finite over the entire real axis. Then, by applying the Cauchy-Schwarz inequality for random variables, we can write
\begin{equation}
M(s_1,s_2)\leq \sqrt{M_1(2 s_1)\,M_2(2 s_2)},
\end{equation}
which shows that the bivariate MGF exists and is finite for all $(s_1,s_2)\in\mathbb{R}^2$.
Moreover, by exploiting the properties of the exponential function and the dominated convergence theorem, it is possible to show that the MGF is infinitely differentiable over $\mathbb{R}^2$ and that its derivatives can be computed by exchanging the differentiation and integration operators, which yields, in particular, for $n=1,2$,~\cite{BillingsleyProbMeas}
\begin{equation}
\frac{\partial M(s_1,s_2)}{\partial s_n}=
\mathbb{E}_{\thetatrue}\, \left
[\bm{v}_n\,e^{
s_1\,\bm{v}_1
+
s_2\,\bm{v}_2
}
\right]
\label{eq:mgfirstder}
\end{equation}
and, likewise, for all pairs $m,n$, including $m=n$,
\begin{equation}
\frac{\partial M(s_1,s_2)}{\partial s_m\partial s_n}=
\mathbb{E}_{\thetatrue}\, \left
[\bm{v}_m\, \bm{v}_n
\,e^{
s_1\,\bm{v}_1
+
s_2\,\bm{v}_2
}
\right].
\label{eq:mgfsecder}
\end{equation}
Since $\Lambda(s_1,s_2)=\ln M(s_1,s_2)$, from \eqref{eq:mgfirstder} we obtain 
\begin{equation}
\frac{\partial \Lambda(s_1,s_2)}{\partial s_n}=
\frac{\mathbb{E}_{\thetatrue}\, \left
[\bm{v}_n\,e^{
s_1\,\bm{v}_1
+
s_2\,\bm{v}_2
}
\right]}
{M(s_1,s_2)}.
\label{eq:Lambda12firstder}
\end{equation}
Referring to the global likelihood $\ell(x|\thetatrue)$ defined by \eqref{eq:globlikeli}, we introduce the tilted measure 
\begin{equation}
\tilde{\ell}(x|\thetatrue)=\frac
{e^{s_1 v_1+s_2 v_2}
}
{
e^{\Lambda(s_1,s_2)}
}
\,\ell(x|\thetatrue),
\label{eq:tilt2dim}
\end{equation}
where we note that the variables $v_1$ and $v_2$ depend implicitly on $x$ since they are log likelihood ratios; see \eqref{eq:v1v2def}.
Using \eqref{eq:tilt2dim} in \eqref{eq:Lambda12firstder}, we conclude that
\begin{equation}
\frac{\partial \Lambda(s_1,s_2)}{\partial s_n}
=\mathbb{E}_{\tilde{\ell}_{\thetatrue}} \bm{v}_n,
\label{eq:meantilt2dim}
\end{equation}
where the subscript denotes that the expectation is computed under the tilted pdf \eqref{eq:tilt2dim}. 

In order to establish the strict convexity of $\Lambda(s_1,s_2)$, we will prove that the associated Hessian matrix is positive definite. Let us accordingly evaluate the second partial derivatives of $\Lambda(s_1,s_2)$. Using again the identity $\Lambda(s_1,s_2)=\ln M(s_1,s_2)$, we have, for all pairs $m,n$, including $m=n$, 
\begin{equation}
\frac{\partial \Lambda(s_1,s_2)}{\partial s_m\partial s_n}=
\frac{\dfrac{\partial M(s_1,s_2)}{\partial s_m\partial s_n}}{M(s_1,s_2)}\nonumber\\
-
\frac{
\dfrac{\partial M(s_1,s_2)}{\partial s_m}\dfrac{\partial M(s_1,s_2)}{\partial s_n}
}{M^2(s_1,s_2)},
\label{eq:lmgfsecder0}
\end{equation}
which, exploiting \eqref{eq:mgfirstder}--\eqref{eq:meantilt2dim}, can be cast in the form
\begin{equation}
\frac{\partial \Lambda(s_1,s_2)}{\partial s_m\partial s_n}=
\mathbb{E}_{\tilde{\ell}_{\thetatrue}}[\bm{v}_m\,\bm{v}_n] - 
\mathbb{E}_{\tilde{\ell}_{\thetatrue}}[\bm{v}_m]\,\mathbb{E}_{\tilde{\ell}_{\thetatrue}}[\bm{v}_n]. 
\label{eq:basicsecder}
\end{equation}
By introducing the vector
\begin{equation}
\bm{v}\triangleq \big[
\bm{v}_1,\,\bm{v}_2
\big]^\top
\end{equation}
and the centered (under the tilted measure) version thereof,
\begin{equation}
\bar{\bm{v}}=
\bm{v} - \nabla\Lambda(s_1,s_2),
\end{equation}
from \eqref{eq:basicsecder} we get the following representation for the Hessian matrix of the bivariate LMGF:
\begin{equation}
\nabla^2 \Lambda(s_1,s_2)=
\mathbb{E}_{\tilde{\ell}_{\thetatrue}} \left[
\bar{\bm{v}}_t\,\bar{\bm{v}}_t^{\top}
\right].
\end{equation}
That is, the Hessian matrix is the covariance matrix of $\bar{\bm{v}}_t$. Now, it is known that the covariance matrix is positive definite unless the vector $\bar{\bm{v}}_t$ has at least one deterministic component, or the two components are proportional. However, both cases are impossible in view of Lemma~\ref{lem:LLRnondeg}, and the strict convexity of $\Lambda(s_1,s_2)$ follows. 
\end{IEEEproof}

\begin{theorem}[\textbf{Exponents for joint threshold crossing}]
\label{th:jointhreshcross}
Under Assumptions~\ref{assum:luckylike}--\ref{assum:finiteMGFs}, consider the true hypothesis $\thetatrue$ and two other hypotheses $\theta_1$ and $\theta_2$ such that
\begin{equation}
\mathcal{E}(\thetatrue,\theta_1)=\mathcal{E}(\thetatrue,\theta_2).
\end{equation}
Then,\footnote{$o(t)$ is a function satisfying the condition $\lim\limits_{t\rightarrow\infty}o(t)/t=0$.}
\begin{equation}
\mathbb{P}\left[
\bm{\beta}_{k,t}(\thetatrue,\theta_1)\leq 0\, , \bm{\beta}_{k,t}(\thetatrue,\theta_2)\leq 0
\right]
=e^{
- t \mathcal{E}(\thetatrue,\theta_1,\theta_2) + o(t)
},
\label{eq:jointerrexplemclaim}
\end{equation}
with $\mathcal{E}(\thetatrue,\theta_1,\theta_2)>\mathcal{E}(\thetatrue,\theta_1)$.
\end{theorem}
\begin{IEEEproof}
We start by computing the error exponent for joint threshold crossing, $\mathcal{E}(\thetatrue,\theta_1,\theta_2)$, which appears in \eqref{eq:jointerrexplemclaim}. To this aim, we will apply the G{\"a}rtner-Ellis theorem~\cite{DemboZeitouni,DenHollander}. First, we need to introduce the bivariate LMGF of the \emph{scaled} log belief ratios 
\begin{equation}
\frac 1 t \, \bm{\beta}_{k,t}(\thetatrue,\theta_1),\quad
\frac 1 t \, \bm{\beta}_{k,t}(\thetatrue,\theta_2),
\end{equation}
which is given by~\cite{DemboZeitouni,DenHollander}
\begin{align}
\Lambda_{\beta}&(s_1,s_2;\theta_0,\theta_1,\theta_2)
\nonumber\\
&\triangleq
\ln \mathbb{E}_{\thetatrue} \exp
\Big(
\frac{s_1}{t}\,\bm{\beta}_{k,t}(\thetatrue,\theta_1)
+
\frac{s_2}{t} \,\bm{\beta}_{k,t}(\thetatrue,\theta_2)
\Big).
\end{align}
The G{\"a}rtner-Ellis theorem prescribes to compute the following limit value for the LMGF:
\begin{equation}
\lim_{t\rightarrow\infty} \frac 1 t\, \Lambda_{\beta}(s_1\,t,s_2\,t;\theta_0,\theta_1,\theta_2).
\label{eq:limLMGFnectoGEll}
\end{equation}
Applying the definition of the log belief ratios in \eqref{eq:betaktinitdef}, we can write
\begin{align}
&\frac 1 t\, \Lambda_{\beta}(s_1\,t,s_2\,t;\theta_0,\theta_1,\theta_2)
\nonumber\\
&=
\frac 1 t
\left(
s_1\,\ln\frac{\pi(\thetatrue)}{\pi(\theta_1)}
+
s_2\,\ln\frac{\pi(\thetatrue)}{\pi(\theta_2)}
\right)
\nonumber\\
&+\frac 1 t
\ln\mathbb{E}_{\thetatrue}\,e^{
\sum_{j=1}^K\sum_{\tau=1}^{t-d_{jk}} 
\big(
s_1 \bm{\lambda}_{j,\tau}(\thetatrue,\theta_1)
+ 
s_2 \bm{\lambda}_{j,\tau}(\thetatrue,\theta_2)
\big)}
\nonumber\\
&=
\frac 1 t 
\left(
s_1\,\ln\frac{\pi(\thetatrue)}{\pi(\theta_1)}
+
s_2\,\ln\frac{\pi(\thetatrue)}{\pi(\theta_2)}
\right)
\nonumber\\
&+
\frac 1 t 
\sum_{j=1}^K\sum_{\tau=1}^{t-d_{jk}} 
\!\!\ln\mathbb{E}_{\thetatrue}\!\exp\Big(
s_1 \bm{\lambda}_{j,\tau}(\thetatrue,\theta_1)
+ 
s_2 \bm{\lambda}_{j,\tau}(\thetatrue,\theta_2)
\Big),
\label{eq:LMGFbetanewlastapp}
\end{align}
where the last step follows from the independence over space and time. We remark that the expectations are finite in view of Assumption~\ref{assum:finiteMGFs} and the Cauchy-Schwarz inequality for random variables.
The first term on the RHS in \eqref{eq:LMGFbetanewlastapp} vanishes as $t\rightarrow\infty$. We focus next on the second term. Since the distances $d_{jk}$ are finite, we can equivalently compute the limit of the quantity
\begin{align}
&\frac 1 t 
\sum_{j=1}^K\sum_{\tau=1}^{t} 
\ln\mathbb{E}_{\thetatrue}\exp\Big(
s_1 \bm{\lambda}_{j,\tau}(\thetatrue,\theta_1)
+ 
s_2 \bm{\lambda}_{j,\tau}(\thetatrue,\theta_2)
\Big)
\nonumber\\
&=
\frac 1 t \sum_{\tau=1}^{t} \ln\mathbb{E}_{\thetatrue}e^{
\sum_{j=1}^K
\big(
s_1 \,\bm{\lambda}_{j,\tau}(\thetatrue,\theta_1)
+ 
s_2 \,\bm{\lambda}_{j,\tau}(\thetatrue,\theta_2)
\big)
}\nonumber\\
&=
\frac 1 t \sum_{\tau=1}^{t} \ln\mathbb{E}_{\thetatrue}\exp
\Big(
s_1 \,\bm{\lambda}_{\tau}(\thetatrue,\theta_1)
+ 
s_2 \,\bm{\lambda}_{\tau}(\thetatrue,\theta_2)
\Big)
\nonumber\\
&=
\frac 1 t \sum_{\tau=1}^{t} \Lambda
(s_1,s_2;\thetatrue,\theta_1,\theta_2
)=\Lambda
(s_1,s_2;\thetatrue,\theta_1,\theta_2
),
\end{align}
where: in the first equality we exploited the independence over space; in the second step we applied the definition the aggregate log likelihood ratios from \eqref{eq:sumLLRvariable}; and in the third equality we used the definition of the bivariate LMGF for the aggregate log likelihood ratios from \eqref{eq:bivLMGFlikedef}. 

Once the limiting LMGF in \eqref{eq:limLMGFnectoGEll} exists, the G\"{a}rtner-Ellis theorem ensures the exponential behavior for the probability described by  \eqref{eq:jointerrexplemclaim}, with the following error exponent:
\begin{equation}
\mathcal{E}(\thetatrue,\theta_1,\theta_2)=
\inf_{y_1\leq 0, y_2\leq 0}\Lambda^{\star}(y_1,y_2;\thetatrue,\theta_1,\theta_2),
\label{eq:ratefunlemdef}
\end{equation}
where $\Lambda^{\star}(y_1,y_2;\thetatrue,\theta_1,\theta_2)$ is the Fenchel-Legendre transform of $\Lambda(s_1,s_2;\thetatrue,\theta_1,\theta_2)$, defined as~\cite{DemboZeitouni,DenHollander}:\footnote{We remark that the function in \eqref{eq:FLbivdef} is defined as an extended real number\cite{DemboZeitouni}.}
\begin{align}
    \Lambda^{\star}&(y_1,y_2;\thetatrue,\theta_1,\theta_2) = \nonumber\\
    &\sup_{s_1,s_2 \in \mathbb{R}^2}\, \big[(s_1 y_1 + s_2 y_2) - \Lambda(s_1,s_2;\thetatrue,\theta_1,\theta_2) \big].
    \label{eq:FLbivdef}
\end{align}
From \cite[Lemma 2.2.31]{DemboZeitouni}, we know that the rate function is nonnegative, convex, and that
\begin{equation}
\Lambda^{\star}(m_1,m_2;\thetatrue,\theta_1,\theta_2)=0,
\label{eq:zerorate}
\end{equation}
where
\begin{equation}
m_1=\mathbb{E}_{\thetatrue}\bm{\lambda}_{\tau}(\thetatrue,\theta_1)>0,\qquad
m_2=\mathbb{E}_{\thetatrue}\bm{\lambda}_{\tau}(\thetatrue,\theta_2)>0.
\end{equation}
(Recall that the positivity of the means follows from Assumption~\ref{assum:ident}).
To compute the error exponent in \eqref{eq:ratefunlemdef}, we need to evaluate the infimum of the rate function over the region $\mathbb{R}^2_{-}\triangleq\{y_1\leq 0, y_2\leq 0\}$. 
We now show that it is sufficient to focus on the boundary of this region, \begin{equation}
\partial \mathbb{R}^2_{-}=\{y_1\leq 0, y_2=0\}\cup\{y_1=0, y_2\leq 0\}.
\end{equation}
Consider a point $y^{o}$ internal to $\mathbb{R}^2_{-}$, i.e., $y^{o}=[y^{o}_1,y^{o}_2]^{\top}\in \mathbb{R}^2_{-}\setminus\partial \mathbb{R}^2_{-}$, and consider the boundary point $y=[y_1,y_2]^{\top}\in \partial \mathbb{R}^2_{-}$ that lies over the chord joining $y^{o}$ to $m=[m_1,m_2]^{\top}$. For some $0<\alpha<1$, the point $y$ can be represented as
\begin{equation}
y=\alpha \, y^{o} + (1-\alpha) \, m.
\end{equation}
From the convexity of $\Lambda^{\star}(y_1,y_2;\thetatrue,\theta_1,\theta_2)$, we have that
\begin{align}
&\Lambda^{\star}(y_1,y_2;\thetatrue,\theta_1,\theta_2)\nonumber\\
&\leq
\alpha \Lambda^{\star}(y^o_1,y^o_2;\thetatrue,\theta_1,\theta_2)+
(1-\alpha)\underbrace{\Lambda^{\star}(m_1,m_2;\thetatrue,\theta_1,\theta_2)}_{=0\textnormal{ in view of \eqref{eq:zerorate}}}\nonumber\\
&\leq
\Lambda^{\star}(y^o_1,y^o_2;\thetatrue,\theta_1,\theta_2),
\label{eq:ratefunineqcvx}
\end{align}
where the last step follows from the fact that $0<\alpha<1$ and from the nonnegativity of the rate function. 
Using \eqref{eq:ratefunineqcvx} in \eqref{eq:ratefunlemdef}, we conclude that
\begin{equation}
\mathcal{E}(\thetatrue,\theta_1,\theta_2)=
\inf_{y\in\partial\mathbb{R}^2_{-}}\Lambda^{\star}(y_1,y_2;\thetatrue,\theta_1,\theta_2).
\end{equation}
In \cite[Lemma 2.2.31]{DemboZeitouni}, it is also established that $\Lambda^{\star}(y_1,y_2;\thetatrue,\theta_1,\theta_2)$ is a ``good'' rate function. This means~\cite[p. 4]{DemboZeitouni} that its level sets are compact subsets of $\mathbb{R}^2$, which in turn implies that the function is coercive, namely, it diverges to $+\infty$ as $\|y\|\rightarrow + \infty$. As a result, the infimum is in fact a minimum attained at some point $y\in\partial\mathbb{R}^2_{-}$. Assume that the minimum is attained at $[\bar{y}_1,0]^{\top}$, for some $\bar{y}_1 \leq 0$ (the following arguments apply even in the complementary case where the minimum is attained at $[0, \bar{y}_2]^{\top}$, for some $\bar{y}_2 \leq 0$). 
Then we can write
\begin{align}
\mathcal{E}(\thetatrue,\theta_1,\theta_2)&=
\Lambda^{\star}(\bar{y}_1,0;\thetatrue,\theta_1,\theta_2)\nonumber\\
&=\sup_{(s_1,s_2)\in\mathbb{R}^2} 
\big[
s_1\,\bar{y}_1 - \Lambda(s_1,s_2;\thetatrue,\theta_1,\theta_2)
\big]
\nonumber\\
&\geq
\sup_{s_1\in\mathbb{R}} 
\big[
s_1\,\bar{y}_1 - \Lambda(s_1,0;\thetatrue,\theta_1,\theta_2)
\big]
\nonumber\\
&=
\sup_{s_1\in\mathbb{R}} 
\big[
s_1\,\bar{y}_1 - \Lambda(s_1;\thetatrue,\theta_1)
\big]\!=\!\Lambda^{\star}(\bar{y}_1;\thetatrue,\theta_1),
\label{eq:ineqchainmultiLMGF}
\end{align}
where the second-to-last step follows from the fact that the bivariate LMGF is equal to the univariate LMGF whenever one of the two arguments is zero, while the last step follows from the definition of the rate function in the univariate case. 
Since the univariate rate function $\Lambda^\star(y_1;\thetatrue,\theta_1)$ is minimized at $m_1 > 0$~\cite[Lemma 2.2.5]{DemboZeitouni}, if $\bar{y}_1<0$, the strict convexity of $\Lambda^\star(y_1;\thetatrue,\theta_1)$~\cite[Lemma E.1]{MattaBordignonSayedBook} implies that
\begin{equation}
\Lambda^{\star}(\bar{y}_1;\thetatrue,\theta_1)>\Lambda^{\star}(0;\thetatrue,\theta_1)=\mathcal{E}(\thetatrue,\theta_1),
\end{equation}
and the claim of the theorem is proved. 
When $\bar{y}_1=0$, from \eqref{eq:ineqchainmultiLMGF}, we know that
\begin{align}
\sup_{(s_1,s_2)\in\mathbb{R}^2} 
\big[
- \Lambda(s_1,s_2;\thetatrue,\theta_1,\theta_2)
\big]
&\geq 
\sup_{s_1\in\mathbb{R}} 
\left[
- \Lambda(s_1;\thetatrue,\theta_1)
\right]\nonumber\\
&=\mathcal{E}(\thetatrue,\theta_1).
\label{eq:blabla1}
\end{align}
Recall that the univariate LMGF $\Lambda(s_1;\thetatrue,\theta_1)$ admits a unique minimizer $s(\thetatrue,\theta_1)$. 
Thus, if the relation in \eqref{eq:blabla1} were satisfied with equality, we would conclude that the function $\Lambda(s_1,s_2;\thetatrue,\theta_1,\theta_2)$ admits a minimum located at the point $(s(\thetatrue,\theta_1),0)$. 

Similarly, since 
\begin{equation}
\Lambda(0,s_2;\thetatrue,\theta_1,\theta_2)
=
\Lambda(s_2;\thetatrue,\theta_2),
\end{equation}
by exchanging $s_1$ with $s_2$ and $\theta_1$ with $\theta_2$ in \eqref{eq:ineqchainmultiLMGF}, we have
\begin{align}
\sup_{(s_1,s_2)\in\mathbb{R}^2} 
\big[
- \Lambda(s_1,s_2;\thetatrue,\theta_1,\theta_2)
\big]
&\geq 
\sup_{s_2\in\mathbb{R}} 
\left[
- \Lambda(s_2;\thetatrue,\theta_2)
\right]\nonumber\\
&=\mathcal{E}(\thetatrue,\theta_2),
\label{eq:blabla2}
\end{align}
and since $\mathcal{E}(\thetatrue,\theta_1)=\mathcal{E}(\thetatrue,\theta_2)$ by assumption, we would conclude that the bivariate LMGF has another minimizer located at the point $(0,s(\thetatrue,\theta_2))$, which is distinct from $(s(\thetatrue,\theta_1),0)$ since both $s(\thetatrue,\theta_1)$ and $s(\thetatrue,\theta_2)$ are nonzero. On the other hand, the bivariate LMGF is strictly convex in view of Lemma~\ref{lem:LMGFsc}, and, thus, it cannot have two distinct minimizers, which in turn implies that the equalities in \eqref{eq:blabla1} and \eqref{eq:blabla2} cannot hold, and the proof is complete.
\end{IEEEproof}

\bibliographystyle{IEEEtran}

\begin{thebibliography}{10}
\providecommand{\url}[1]{#1}
\csname url@samestyle\endcsname
\providecommand{\newblock}{\relax}
\providecommand{\bibinfo}[2]{#2}
\providecommand{\BIBentrySTDinterwordspacing}{\spaceskip=0pt\relax}
\providecommand{\BIBentryALTinterwordstretchfactor}{4}
\providecommand{\BIBentryALTinterwordspacing}{\spaceskip=\fontdimen2\font plus
\BIBentryALTinterwordstretchfactor\fontdimen3\font minus
  \fontdimen4\font\relax}
\providecommand{\BIBforeignlanguage}[2]{{%
\expandafter\ifx\csname l@#1\endcsname\relax
\typeout{** WARNING: IEEEtran.bst: No hyphenation pattern has been}%
\typeout{** loaded for the language `#1'. Using the pattern for}%
\typeout{** the default language instead.}%
\else
\language=\csname l@#1\endcsname
\fi
#2}}
\providecommand{\BIBdecl}{\relax}
\BIBdecl

\bibitem{Veeravalli}
J.~F. Chamberland and V.~V. Veeravalli, ``Decentralized detection in sensor networks,'' \emph{IEEE Trans. Signal Process.}, vol.~51, no.~2, pp. 407--416, Feb. 2003.

\bibitem{Varshney}
R.~Viswanathan and P.~K. Varshney, ``Distributed detection with multiple sensors Part I. Fundamentals,'' \emph{Proceedings of the IEEE}, vol.~85, no.~1, pp. 54--63, Jan. 1997.

\bibitem{Poor}
R.~S. Blum, S.~A. Kassam, and H.~V. Poor, ``Distributed detection with multiple sensors II. Advanced topics,'' \emph{Proceedings of the IEEE}, vol.~85, no.~1, pp. 64--79, Jan. 1997.

\bibitem{Tong}
G.~Mergen, V.~Naware, and L.~Tong, ``Asymptotic detection performance of type-based multiple access over multiaccess fading channels,'' \emph{IEEE Trans. Signal Process.}, vol.~55, no.~3, pp. 1081--1092, Mar. 2007.

\bibitem{Willett}
S.~Marano, V.~Matta, L. Tong, and P.~Willett, ``A likelihood-based multiple access for estimation in sensor networks,'' \emph{IEEE Trans. Signal Process.}, vol.~55, no.~11, pp. 5155--5166, Nov. 2007.

\bibitem{MattaBordignonSayedBook}
 V.~Matta, V.~Bordignon, and A.~H. Sayed, \emph{Social {{Learning}}: {{Opinion
   Formation}} and {{Decision-Making Over Graphs}}}.\hskip 1em plus 0.5em minus
   0.4em\relax Emerald Publishing, 2025.

\bibitem{SocLearnSPmagazine}
V.~Bordignon, V.~Matta, and A.~H. Sayed, ``Socially {{intelligent networks}}:
  {{A}} framework for decision making over graphs,'' \emph{IEEE Sig.
  Process. Mag.}, vol.~41, no.~4, pp. 20--39, Jul. 2024.

\bibitem{zhaoLearningSocialNetworks2012}
X.~Zhao and A.~H. Sayed, ``Learning over social networks via diffusion adaptation,'' in \emph{Proc. Asilomar Conference on Signals, Systems, and Computers}, Pacific Grove, CA, USA, 2012, pp. 709--713.
 
\bibitem{jadbabaieNonBayesianSocialLearning2012}
A.~Jadbabaie, P.~Molavi, A.~Sandroni, and A.~{Tahbaz-Salehi},
  ``Non-{{Bayesian}} social learning,'' \emph{Games and Economic Behavior},
  vol.~76, no.~1, pp. 210--225, Sep. 2012.

\bibitem{lalithaSocialLearningDistributed2018}
A.~Lalitha, T.~Javidi, and A.~D. Sarwate, ``Social learning and distributed hypothesis testing,'' \emph{IEEE Trans. Inf. Theory}, vol.~64, no.~9, pp. 6161--6179, Sep. 2018.
 
\bibitem{nedicFastConvergenceRates2017}
A.~Nedi{\'c}, A.~Olshevsky, and C.~A. Uribe, ``Fast convergence rates for distributed non-Bayesian learning,'' \emph{IEEE Trans. Automat. Control}, vol.~62, no.~11, pp. 5538--5553, Nov. 2017.

\bibitem{Jadbabaie2018}
P.~Molavi, A.~Tahbaz-Salehi, and A.~Jadbabaie, ``A theory of non-{B}ayesian social learning,'' \emph{Econometrica}, vol.~86, no.~2, pp. 445--490, Mar. 2018.

\bibitem{MouraLDnonGauss}
D.~Bajović, D.~Jakovetic, J.~M.~F. Moura, J.~Xavier, and B.~Sinopoli, ``Large deviations performance of consensus+innovations distributed detection with non-{G}aussian observations,'' \emph{{IEEE} Trans. Signal Process.}, vol.~60, no.~11, pp. 5987--6002, Nov. 2012.

\bibitem{KayDet}
S.~M. Kay, \emph{Fundamentals of Statistical Signal Processing: Detection Theory}.\hskip 1em plus 0.5em minus 0.4em\relax Prentice Hall, NJ, 1998.

\bibitem{chenSayedLearningBehavior2015PartI}
J.~Chen and A.~H. Sayed, ``On the learning behavior of adaptive networks—Part I: Transient analysis,'' \emph{IEEE Trans. Inf. Theory}, vol.~61, no.~6, pp. 3487--3517, Jun. 2015.

\bibitem{chenSayedLearningBehavior2015PartII}
J.~Chen and A.~H. Sayed, ``On the learning behavior of adaptive networks—Part II: Performance analysis,'' \emph{IEEE Trans. Inf. Theory}, vol.~61, no.~6, pp. 3518--3548, Jun. 2015.

\bibitem{DemboZeitouni}
A.~Dembo and O.~Zeitouni, \emph{Large Deviations Techniques and Applications}.\hskip 1em plus 0.5em minus 0.4em\relax Springer, 2009.
 
\bibitem{DenHollander}
F.~{den Hollander}, \emph{Large {{Deviations}}}.\hskip 1em plus 0.5em minus
  0.4em\relax AMS, 2000.

\bibitem{HuangWang}
B.~Huang, I-H.~Wang ``On the price of decentralization in decentralized detection,'' \emph{IEEE Trans. Inf. Theory}, vol.~71, no.~4, pp. 2341--2359, Apr. 2025.

\bibitem{ourEUSIPCO2025}
 F.~Scala, M.~Carpentiero, V.~Matta and A.~H. Sayed, ``On the performance of social learning,'' in \emph{Proc. EUSIPCO}, Palermo, Italy, Sep. 2025, pp. 1040--1044.

\bibitem{ourEUSIPCOpaperarxiv2025}
 F.~Scala, M.~Carpentiero, V.~Matta and A.~H. Sayed, ``Performance evaluation of social learning,'' \emph{accepted for publication in IEEE Trans. Sig. Process.}, Jun. 2026, available online at arXiv:2606.09176 [cs.MA].

\bibitem{CoverThomas}
T.~M. Cover and J.~A. Thomas, \emph{Elements of Information Theory}. \hskip 1em plus 0.5em minus
   0.4em\relax Wiley, 1991.

\bibitem{HornJohnson}
 R.~A. Horn and C.~R. Johnson, \emph{Matrix Analysis}. Cambridge {U}niversity {P}ress, 2012.

\bibitem{ErdosRenyi}
P.~Erd\H{o}s and A.~R\'{e}nyi, ``On random graphs I,'' \emph{Publicationes Mathematicae (Debrecen)}, vol.~6, pp. 290--297, 1959.

\bibitem{RGbook}
 B.~Bollob\'as, \emph{Random Graphs}. Cambridge {U}niversity {P}ress, 2001.

\bibitem{APLphysrevE}
A.~Fronczak, P.~Fronczak, and J. A.~Ho{\l}yst,
   ``Average path length in random networks,'' \emph{Phys. Rev. E},
   vol.~70, no.~056110, pp. 056110-1--056110-7, 2004.

\bibitem{BillingsleyProbMeas}
 P.~Billingsley, \emph{Probability and Measure}.\hskip 1em plus 0.5em minus
   0.4em\relax Wiley, 1995.

\bibitem{Petrov}
 V.~V. Petrov, \emph{Sums of Independent Random Variables}.\hskip 1em plus 0.5em minus
   0.4em\relax Springer, 1975.
   
\bibitem{Feller2}
W.~Feller, \emph{An Introduction to Probability and Its Applications, {\rm
  vol.} 2},\hskip 1em plus 0.5em minus 0.4em\relax Wiley, 1971.

\bibitem{GalambosBonferroniBook}
J.~Galambos and I.~Simonelli, \emph{Bonferroni-Type Inequalities with Applications}.\hskip 1em plus 0.5em minus
  0.4em\relax Springer-Verlag, 1996.



\end{thebibliography}


\end{document}